\makeatletter
\let\@twosidetrue\@twosidefalse
\let\@mparswitchtrue\@mparswitchfalse
\makeatother

\documentclass{llncs}

\usepackage{amsmath}
\usepackage{latexsym,amsfonts}
\usepackage{amssymb}
\usepackage{array}
\usepackage{graphicx,color,enumerate,paralist,lmodern}
\usepackage{bbm}

\usepackage{url}
\usepackage[textsize=tiny,textwidth=2.9cm,shadow]{todonotes}
\newcommand{\wt}{\mathsf{wt}}
\newcommand{\vote}{\mathsf{vote}}
\newtheorem{new-claim}{Claim}
\newcommand{\cupdot}{\mathbin{\mathaccent\cdot\cup}}
\newcommand{\true}{\mathsf{true}}
\newcommand{\false}{\mathsf{false}}
\newtheorem{pr}{Problem}
\newcommand{\inp}{\textsf{Input: }} 

\newcommand{\ques}{\textsf{Decide: }}
\newcommand{\new}[1]{\textcolor{black}{#1}}

\definecolor{teal}{RGB}{0,128,128} 
\usepackage[breaklinks=true]{hyperref} 
\hypersetup{
	colorlinks=true,
	citecolor=teal,
	linkcolor=blue,  
}

\usepackage{tikz}
\usetikzlibrary{lindenmayersystems,calc,decorations,decorations.pathmorphing,decorations.markings,shapes,shapes.geometric,arrows, backgrounds, patterns,decorations.pathreplacing, positioning}
\pgfdeclarelayer{background}
\pgfdeclarelayer{foreground}
\pgfsetlayers{background,main,foreground}
\definecolor{MyPurple}{RGB}{111,111,111}
\tikzstyle{uvertex} = [circle, draw=black, fill=black, inner sep=0pt,  minimum size=5pt]
\tikzstyle{vvertex} = [circle, draw=black, fill=white, inner sep=0pt,  minimum size=5pt]
\tikzstyle{dot} = [circle, draw=black, fill=black, inner sep=0pt,  minimum size=2pt]
\tikzstyle{edgelabel} = [circle, fill=white, inner sep=0pt,  minimum size=10pt]
\tikzstyle{edgelabelr} = [rectangle, fill=white, inner sep=0pt,  minimum size=12pt]
\tikzstyle{edgebox} = [inner sep=3pt, draw=black, fill = white,  minimum size=14pt]
\usetikzlibrary{decorations.pathmorphing}
\tikzset{snake it/.style={decorate, decoration=snake}}

\begin{document}
\title{Understanding popular matchings via stable matchings}
\author{\'Agnes Cseh\inst{1,2} \and Yuri Faenza\inst{3} \and Telikepalli Kavitha\inst{4}\thanks{Part of this work was done while visiting MPI for Informatics, Saarland Informatics Campus, Germany.} \and Vladlena Powers\inst{3}}
\institute{Hasso-Plattner-Institute, University of Potsdam, Germany \and Institute of Economics, Centre for Economic and Regional Studies, Hungary; \email{cseh.agnes@krtk.hu} \and IEOR, Columbia University, New York, USA;
	\email{\{yf2414, vp2342\}@columbia.edu} \and Tata Institute of Fundamental Research, Mumbai, India; \email{kavitha@tifr.res.in}}
\maketitle
\pagestyle{plain}

\begin{abstract}
An instance of the marriage problem is given by a graph $G = (A \cup B,E)$, together with, for each vertex of $G$, a strict preference order over its neighbors. A matching $M$ of $G$ is {\em popular} in the marriage instance if $M$ does not lose a head-to-head election against any matching where vertices are voters. Every stable matching is a {\em min-size} popular matching; another subclass of popular matchings that always exist and can be easily computed is the set of {\em dominant} matchings. A popular matching $M$ is dominant if $M$ wins the head-to-head election against any larger matching. Thus every dominant matching is a {\em max-size} popular matching and it is known that the set of dominant matchings is the linear image of the set of stable matchings in an auxiliary graph. Results from the literature  seem to suggest that stable and dominant matchings behave, from a complexity theory point of view, in a very similar manner within the class of popular matchings.

The goal of this paper is to show that there are instead differences in the tractability of stable and dominant matchings, and to investigate further their importance for popular matchings. First, we show that it is easy to check if all popular matchings are also stable, however it is co-NP hard to check if all popular matchings are also dominant. Second, we show how some new and recent hardness results on \emph{popular} matching problems can be deduced from the NP-hardness of certain problems on \emph{stable} matchings, also studied in this paper, thus showing that stable matchings can be employed not only to show positive results on popular matching (as is known), but also most negative ones. Problems for which we show new hardness results include finding a min-size (resp.~max-size) popular matching that is not stable (resp.~dominant).  A known result for which we give a new and simple proof is the NP-hardness of finding a popular matching when $G$ is non-bipartite.
\end{abstract}

\section{Introduction}
\label{sec:intro}
Consider a bipartite graph $G = (A \cup B,E)$ on $n$ vertices and $m$ edges where each vertex has a strict ranking of its neighbors. Such a graph \new{supplied with preference lists}, also called a {\em marriage} instance, is an extensively studied model in two-sided matching markets. The problem of computing a 
{\em stable} matching in $G$ is classical. A matching $M$ is stable if there is no {\em blocking edge} with respect to $M$, i.e., an edge whose endpoints prefer each other to their respective assignments in $M$. The notion of stability was introduced by Gale and Shapley~\cite{GS62} in 1962 who showed that stable matchings always exist in $G$ and there is a simple linear time algorithm to find one. 

Stable matchings in an instance \new{with an underlying bipartite graph} $G$ are well-understood~\cite{GI89}, with efficient algorithms~\cite{Fed92,Fed94,ILG87,Rot92,TS98,VV89} to solve several optimization problems that have many applications in economics, computer science, mathematics, and operations research. Here we study a related and more relaxed notion called {\em popularity}. This notion was introduced by G\"ardenfors~\cite{Gar75} in 1975 who
showed that every stable matching is also popular. For any vertex $u$, its preference over neighbors extends naturally to a preference over matchings as follows:
$u$ prefers $M$ to $M'$ either if (i)~$u$ is matched in $M$ and unmatched in $M'$ or (ii)~$u$ is matched in both and prefers
its partner in $M$ to its partner in~$M'$. Let $\psi(M,M')$ be the number of vertices that prefer $M$ to~$M'$.

\begin{definition}
\label{pop-def}
	A matching $M$ is {\em popular} if  $\psi(M,M') \ge  \psi(M',M)$ for every matching $M'$ in $G$, 
	i.e., $\Delta(M,M') \ge 0$ where $\Delta(M,M') = \psi(M,M') -  \psi(M',M)$.
\end{definition}

Hence, in a voting-based context, vertices constitute the set of voters, and each matching in the instance is an alternative. In a head-to-head election between two matchings, each vertex casts a vote for the matching that it prefers and it abstains from voting if its assignment is the same in both matchings. A popular matching, by definition, never loses such a head-to-head election against another matching. Equivalently, a popular matching is a weak {\em Condorcet winner}~\cite{Con85,wiki-condorcet} in the corresponding voting instance. It is easy to show that a stable matching is a min-size popular matching~\cite{HK13b}. Thus larger matchings and more generally, matchings that achieve more social good, are possible by relaxing the constraint of stability to popularity.

Algorithmic questions for popular matchings in bipartite graphs have been well-studied in the last decade~\cite{BIM10,CK18,FKPZ18,HK13b,HK17,Kav14,Kav16}. We currently know efficient algorithms for the following problems in bipartite graphs: (i)~min-size popular matchings, (ii)~max-size popular matchings, and (iii)~finding a popular matching with a given edge. All these algorithms compute either a stable matching or a {\em dominant} matching.

\begin{definition}
	\label{def:dominant}
	A popular matching $M$ is dominant in $G$ if $M$ is more popular than any larger matching in $G$, i.e.,  $\Delta(M,M') > 0$ for any matching
	$M'$ such that $|M'| > |M|$.
\end{definition}  

Thus a dominant matching defeats every larger matching in a head-to-head election, so it immediately follows that a dominant matching is a popular matching of maximum size. The example in Fig.~\ref{fig:example} (from \cite{HK13b}) demonstrates the differences between stable, dominant, and max-size matchings. 

In the graph $G = (A \cup B, E)$ here, we have $A = \left\{ a_1, a_2, a_3 \right\}$ and $B = \left\{ b_1, b_2, b_3 \right\}$. The same preferences are depicted as numbers on the edges and as lists to the left of the drawn graph. Vertex $b_1$ is the top choice for all $a_i$'s, $b_2$ is the second choice for $a_1$ and $a_2$, and $b_3$ is the third choice for $a_1$ alone. The preference lists of the $b_i$ vertices are symmetric. There are two popular matchings here, both of them have the same cardinality: $M_1 = \{(a_1,b_1),(a_2,b_2)\}$ and $M_2 = \{(a_1,b_2),(a_2,b_1)\}$. 

The matching $M_1$ is stable, but not dominant, since it is  {\em not} more popular than the larger matching $M_3 = \{(a_1,b_3),(a_2,b_2),(a_3,b_1)\}$. Observe that in an election between $M_1$ and $M_3$, vertices $a_1, b_1$ vote for $M_1$, vertices $a_3, b_3$ vote for $M_3$, and vertices $a_2, b_2$ are indifferent between $M_1$ and $M_3$; thus $\Delta(M_1,M_3) = 2 - 2 = 0$.
The matching $M_2$ is dominant since $M_2$ is more popular than $M_3$: observe that $\Delta(M_2,M_3) = 4 - 2 = 2$ since $a_1,b_1,a_2,b_2$ prefer $M_2$ to $M_3$ while $a_3, b_3$ prefer $M_3$ to $M_2$. However $M_2$ is not stable, since $(a_1,b_1)$ blocks it.

\begin{figure}[h]
\tikzstyle{vertex} = [circle, draw=black, fill=black, inner sep=0pt,  minimum size=5pt]
\tikzstyle{edgelabel} = [circle, fill=white, inner sep=0pt,  minimum size=15pt]
\centering
	\pgfmathsetmacro{\d}{2}
\begin{minipage}{0.4\textwidth}
\[
\begin{array}{llllllll}
a_1 \ : & & b_1 \succ & b_2 \succ & b_3 \ & \hspace*{0.5in}b_1: a_1 \succ & a_2 \succ & a_3\\
a_2 \ : & & b_1 \succ & b_2 \ &       &   \hspace*{0.5in} b_2: a_1 \succ & a_2\  \\
a_3 \ : & & b_1 \ &     \  &      &     \hspace*{0.5in} b_3: a_1
\end{array}
\]
\end{minipage}\hspace{10mm}\begin{minipage}{0.4\textwidth}
\begin{tikzpicture}[scale=1, transform shape]
	\node[vertex, label=above:$a_1$] (a1) at (0,0) {};
	\node[vertex, label=left:$b_2$] (b1) at ($(a1) + (0, -\d)$) {};
	\node[vertex, label=above:$b_1$] (b2) at (\d,0) {};
	\node[vertex, label=right:$a_2$] (a2) at ($(b2) + (0, -\d)$) {};
	\node[vertex, label=above:$b_3$] (b3) at (-\d,0) {};
	\node[vertex, label=above:$a_3$] (a3) at ($(a1) + (2*\d, 0)$) {};
    
	\draw [very thick, teal, snake it] (a1) -- node[edgelabel, near start] {2} node[edgelabel, near end] {1} (b1);
	\draw [very thick, red, dotted] (a1) -- node[edgelabel, near start] {1} node[edgelabel, near end] {1} (b2);
    \draw [very thick] (a1) -- node[edgelabel, near start] {3} (b3);
	\draw [very thick, red, dotted] (a2) -- node[edgelabel, near start] {2} node[edgelabel, near end] {2} (b1);
	\draw [very thick, teal, snake it] (a2) -- node[edgelabel, near start] {1} node[edgelabel, near end] {2} (b2);
    \draw [very thick] (a3) -- node[edgelabel, near end] {3} (b2);
\end{tikzpicture}
\end{minipage}

\caption{The above instance admits two popular matchings. The stable matching $M_1$ is marked by dotted red edges, while the dominant matching $M_2$ is marked by wavy teal edges.}
\label{fig:example}
\end{figure}

Dominant matchings always exist in a bipartite graph~\cite{HK13b} and a dominant matching can be computed in linear time~\cite{Kav14}. Moreover, dominant matchings are the linear image of stable matchings in an auxiliary instance~\cite{CK18}, hence oftentimes an optimization problem over the set of dominant matchings (e.g. finding one of maximum weight) boils down to solving the same problem on the set of stable matchings. Very recently, the following rather surprising result was shown~\cite{FKPZ18}: it is NP-hard to decide if a bipartite graph admits a popular matching that is {\em neither stable nor dominant.}

\subsection{Our problems, results, and techniques}\label{sec:our}

Everything known so far about stable and dominant matchings seemed to suggest that those classes play somehow symmetric roles in popular matching problems in bipartite graphs: both classes are always non-empty and one is a tractable subclass of {\em min-size popular matchings} while the other is a tractable subclass of {\em max-size popular matchings}. Our first set of results shows that this symmetry is not always the case.

Our starting point is an investigation of the complexity of the following two natural and easy-to-ask questions on popular matchings in a bipartite graph $G$ \new{with strict preference lists}: 

\begin{enumerate}[(1)]
    \item is {\em every} popular matching in $G$ also stable?
    \item is {\em every} popular matching in $G$ also dominant?
\end{enumerate}

Both these questions are trivial to answer in instances that admit popular matchings of more than one size.
Then the answer to both questions is  ``no'' since $\{$dominant matchings$\} \cap \{$stable matchings$\} = \emptyset$ in such graphs as
dominant matchings are max-size popular matchings while stable matchings are min-size popular matchings.
Thus, in this case, a dominant matching is an {\em unstable} popular matching and a stable matching is a {\em non-dominant} popular matching in $G$.
However, when all popular matchings in $G$ have the same size, these questions are non-trivial.

Moreover, it is useful to ask these questions because when there are edge utilities, the problem of finding a max-utility
popular matching is NP-hard in general and also hard to approximate to a factor better than 2~\cite{FKPZ18}; however if every popular matching is stable (similarly, dominant), then the max-utility
popular matching problem can be solved in polynomial time. Thus a ``yes'' answer to either of these questions has applications.
We show the following dichotomy here: though both these questions seem analogous, {\em only one} is easy-to-answer. 

\begin{itemize}
    \item[($\ast$)] {\em There is an $O(m^2)$ algorithm to decide if every popular matching in $G = (A \cup B, E)$ is stable, however it is co-NP complete to decide if every popular matching in $G = (A \cup B, E)$ is dominant.}
\end{itemize}

The first step in proving $(\ast)$ is to show that questions (1) and (2) are equivalent to the following:\footnote{Although similar in spirit, the arguments leading from (1) to ($1'$) and from (2) to ($2'$) are not the same, see Lemma~\ref{non-stab-domn} and Lemma~\ref{non-domn-stable}.}

\begin{itemize}
\item[($1'$)] is every \emph{dominant} matching in $G$ also stable?
    \item[($2'$)]\label{it:twoprime} is every \emph{stable} matching in $G$ also dominant?
\end{itemize}

In Section~\ref{sec:dom-vs-stab}, we give a combinatorial algorithm that solves ($1'$) in time $O(m^2)$. We settle the complexity of (2) and ($2'$) in Section~\ref{sec:stable}, showing that the problem is co-NP hard. We deduce the latter from the hardness of finding a stable matching with a certain augmenting path: a result that is shown in this paper.   

\smallskip

Our hardness reduction is surprisingly simple when compared to those that appeared in recent publications on popular matchings~\cite{FKPZ18,GMSZ18,Kav18}, and establishes a new connection between hardness of problems for stable matchings and hardness of problems for popular matchings. This connection turns out to be very fertile: we exploit it further to show NP-hardness of the following new decision problems for a bipartite graph $G$ (in particular, these hardness results are not implied by the reductions from~\cite{FKPZ18,GMSZ18,Kav18}):

\begin{itemize}
    \item[(3)] is there a stable matching in $G$ that is dominant?
    \item[(4)] is there a max-size popular matching in $G$ that is not dominant?
    \item[(5)] is there a min-size popular matching in $G$ that is not stable?
\end{itemize}

  A general graph (not necessarily bipartite) with strict preference lists is called a {\em roommates} instance. 
  Popular matchings need not always exist in a roommates instance and the popular roommates problem is to decide if a given instance admits one or not. The complexity of the popular roommates problem was open for close to a decade and very recently, two independent proofs of NP-hardness~\cite{FKPZ18,GMSZ18} of this problem were shown. Both these proofs are rather lengthy and technical. 
  
  We use the hardness result for problem~(3) to show a short and simple proof of NP-hardness of the popular roommates problem. 
  Moreover, the hardness result for (5) shows an alternative and much simpler proof of NP-hardness (compared to \cite{FKPZ18}) of the following decision problem in a marriage instance $G = (A \cup B, E)$ \new{equipped with strict preference lists}: is there a popular matching in $G$ that is neither stable nor dominant?\footnote{The reduction in \cite{FKPZ18} also showed that it was NP-hard to decide if $G$ admits a popular matching that is neither a min-size nor a max-size popular matching. Our reduction does not imply this.}

Algorithms for computing min-size/max-size popular matchings and for the popular edge problem compute either stable matchings or dominant matchings.
Dominant matchings in $G$ are stable matchings in a related graph $G'$ (see Section~\ref{prelims}) and so the machinery of stable matchings is used to solve dominant matching problems. Thus all positive results in the domain of popular matchings can be attributed to stable matchings. 
Conversely, all hardness results proved in this paper rely on the fact that it is hard to find stable matchings that have / do not have certain augmenting paths. Hence, properties of stable matchings are also responsible for, and provide a unified approach to, the hardness of many popular matching problems\new{.}

\subsection{Background and related results}
In all problems considered in this paper, all vertices of a graph have strict preference lists over their neighbors. The first algorithmic question studied in the domain of popular matchings was in the one-sided preference lists model in bipartite graphs: here, \new{unlike in our setting,} only one side of the graph consists of agents who have preferences over their neighbors; vertices on the other side are objects with no preferences. Popular matchings need not always exist here and an efficient algorithm was shown in \cite{AIKM07} to determine if one exists or not. 

Popular matchings always exist in bipartite graphs when every vertex has a strict preference list~\cite{Gar75}. However when preference lists are not strict, the problem of deciding if a popular matching exists or not is NP-hard~\cite{BIM10,CHK17}. The first non-trivial algorithms designed for computing popular matchings in bipartite graphs with strict preference lists were the max-size popular matching algorithms~\cite{HK13b,Kav14}. These algorithms compute dominant matchings and the term {\em dominant matching} was formally defined a few years later in \cite{CK18} to solve the ``popular edge'' problem.

As mentioned earlier, it was recently shown that it is NP-hard to decide if a  
marriage instance admits a popular matching that is neither stable nor dominant~\cite{FKPZ18}. This hardness result was shown by a reduction from 1-in-3 SAT and a consequence of this hardness result was the hardness of the 
popular roommates problem. The NP-hardness of popular roommates problem shown in \cite{GMSZ18} was by a reduction from a problem called the 
{\em partitioned vertex cover} problem. 

There are several efficient algorithms to solve the stable roommates problem~\cite{Irv85,Sub94,TS98}. In contrast, the {\em dominant roommates}
problem, i.e., the problem of deciding whether a roommates instance admits a dominant matching or not, is NP-hard~\cite{FKPZ18}. Our NP-hardness proof of the popular roommates problem also proves the hardness of the dominant roommates problem and it is much simpler than the proof in \cite{FKPZ18}.

We remark that constructions from the present paper have been employed in the subsequent work~\cite{FK20} to show \emph{polyhedral} results. In particular,~\cite{FK20} builds on Section~\ref{sec:stable} to show that the dominant matching polytope has an exponential number of facets, and it builds on Section~\ref{sec:min-size} to show that the popular matching polytope has near-exponential extension complexity.

\paragraph{Organization of the paper.} Section~\ref{prelims} contains known facts on popular, stable, and dominant matchings, that will be used throughout the paper. In Section~\ref{sec:dom-vs-stab}, we present an algorithm to decide if $G$ has an unstable popular matching.  Section~\ref{sec:stable} has our co-NP hardness result, which is deduced from the problem of deciding whether stable matchings \emph{with} certain augmenting paths exist, whose hardness is also proved in Section~\ref{sec:stable}. In Section~\ref{sec:new} we show how the problem of deciding whether a stable matching \emph{without} certain augmenting paths exists implies new and known hardness results, in particular, the hardness of deciding if there exists a matching that is both stable and dominant. 

\section{Preliminaries}
\label{prelims}

Let $G = (A \cup B, E)$ be \new{the graph in our input}. We will often refer to vertices in $A$ and $B$ as {\em men} and {\em women}, respectively. We will always assume that each vertex has a strict preference order over his/her neighbors\new{, and the set of these lists is denoted by~$\mathcal{P}$. An instance of our problem consists therefore of $G$ and~$\mathcal{P}$, but we will omit to explicitly mention ${\cal P}$ when it is clear from the context.} We often abbreviate $n=|A\cup B|$ and $m=|E|$. We now sketch four important tools developed for popular matchings in earlier papers. Each of these will be used in later parts of this paper.

\medskip

\noindent{\bf 1) Characterization of popular matchings.} Let $M$ be any matching in $G$. For any edge $(a,b) \notin M$, define $\vote_a(b,M)$ as follows (here $M(a)$ is $a$'s partner in the matching
$M$ and $M(a) = \mathsf{null}$ if $a$ is unmatched in $M$):
\begin{equation*} 
\vote_a(b,M) = \begin{cases} +   & \text{if\ $a$\ prefers\ $b$\ to\ $M(a)$};\\
	                     - &  \text{if\ $a$\ prefers\ $M(a)$\ to\ $b$.}			
\end{cases}
\end{equation*}

We can similarly define $\vote_b(a,M)$. Label every edge $(a,b) \notin M$ by $(\vote_a(b,M),\vote_b(a,M))$.  Thus every edge outside
$M$ has a label in $\{(\pm, \pm)\}$. Note that an edge $e$ is labeled $(+,+)$ if and only if $e$ is a blocking edge to $M$.

Let $G_M$ be the subgraph of $G$ obtained by deleting edges labeled $(-,-)$ from $G$. The following theorem  characterizes popular matchings. This characterization holds in non-bipartite graphs as well.

\begin{theorem}[\cite{HK13b}]
  \label{thm:char-popular}
Matching $M$ is popular in instance $G\new{, \mathcal{P}}$ if and only if $G_M$ does not contain any of the following with respect to~$M$:
	\begin{enumerate}
		\item[(i)] an alternating cycle with a $(+,+)$ edge;
		\item[(ii)] an alternating path with two distinct $(+,+)$ edges;
		\item[(iii)] an alternating path with a $(+,+)$ edge and an unmatched vertex as an endpoint.
	\end{enumerate}
\end{theorem}

The following theorem  characterizes dominant matchings.
\begin{theorem}[\cite{CK18}]
\label{thm:dominant}
A popular matching $M$ is dominant iff there is no $M$-augmenting path in $G_M$.
\end{theorem}

\noindent{\bf 2) The graph $G'$.}
Dominant matchings in $G$ are equivalent to stable matchings in a related graph $G'$: this equivalence was first used in \cite{CK18},
later simplified in \cite{FKPZ18}.
The graph $G'$ is the bidirected graph corresponding to $G$. The vertex set of $G'$ is the same as the vertex set $A \cup B$ of $G$
and every edge $(a,b)$ in $G$ is replaced by 2 edges in $G'$: one directed from $a$ to $b$ denoted by $(a^+,b^-)$ and the other directed from $b$ to $a$ denoted by
$(a^-,b^+)$. Let $u \in A\cup B$ and suppose $v_1 \succ v_2 \cdots \succ v_k$ is $u$'s preference order in $G$. Then $u$'s preference order in $G'$ is:
\[ v^-_1 \succ v^-_2 \cdots \succ v^-_k \succ v^+_1 \succ v^+_2 \cdots \succ v^+_k.\]

That is, every vertex prefers outgoing edges to incoming edges and among outgoing edges, it maintains its original preference order and among incoming edges, it maintains again its original preference order. Observe that vertex preferences in $G'$ are expressed on incident edges rather than on neighbors. However it is easy to see that stable matchings in $G'$ are equivalent to stable matchings in the following conventional graph that has 3 vertices $u^+, u^-, d(u)$ corresponding to each vertex $u$ in $G'$. The preference order of $u^+$ is $v^-_1 \succ v^-_2 \cdots \succ v^-_k \succ d(u)$, the preference order of $u^-$ is $d(u) \succ v^+_1 \succ v^+_2 \cdots \succ v^+_k$, and the preference order of $d(u)$ is $u^+ \succ u^-$.

It was shown in \cite{CK18,FKPZ18} that any stable matching $M'$ in $G'$ projects to a dominant matching $M$ in $G$ by setting $(a,b) \in M$ if and only if either $(a^+,b^-)$ or $(a^-,b^+)$ is in $M'$, 
and conversely, any dominant matching in $G$ can be realized as a stable matching in $G'$. 
\medskip

\noindent{\bf 3) The set of matched vertices.} We will be using the {\em Rural Hospitals Theorem} for stable matchings (see e.g.~\cite{GI89}) that states that all stable matchings in $G$ match the same subset of vertices. 
Note that every dominant matching in $G$ matches the same subset of vertices (via the Rural Hospitals Theorem for stable matchings in $G'$). More generally, the following fact is true, where $V(N)$ is the set of vertices matched in a matching $N$.
\begin{lemma}[\cite{Hirakawa-MatchUp15,HK13b}]
\label{lem:all-same}
Let $S$ be a stable, $M$ be a popular, and $D$ be a dominant matching in a marriage instance $G = (A \cup B, E)\new{, \mathcal{P}}$. Then $V(S)\subseteq V(M)\subseteq V(D)$.
\end{lemma}

Thus, in particular, in instances where stable matchings have the same size as dominant matchings, all popular matchings match the same subset of vertices. 

\medskip

\noindent{\bf 4) Witness of a popular matching.}
Let $\tilde{G}$ be the graph $G$ augmented with self-loops. That is, we assume each vertex is its own last choice neighbor. So we can henceforth
regard any matching $M$ in $G$ as a perfect matching $\tilde{M}$ in $\tilde{G}$ by including self-loops at all vertices left unmatched in $M$.
The following edge weight function $\wt_M$ in $\tilde{G}$ will be useful to us. For any edge $(a,b)$ in $G$, define:
\begin{equation*}
    \wt_M(a,b) = \begin{cases}  2 & \text{if\ $(a,b)$\ is\ labeled\ $(+,+)$}\\
                               -2 & \text{if\ $(a,b)$\ is\ labeled\ $(-,-)$}\\
                                0 & \text{otherwise.}
    \end{cases}                            
\end{equation*}

We need to define $\wt_M$ for self-loops as well: let $\wt_M(u,u) = 0$ if $u$ is matched to itself in $\tilde{M}$, else $\wt_M(u,u) = -1$.
For any matching $N$ in $G$, it is easy to see that $\wt_M(\tilde{N}) = \Delta(N,M)$
and so $M$ is popular in $G$ if and only if every perfect matching in $\tilde{G}$ has weight at most 0. Theorem~\ref{thm:witness} follows from
LP-duality and the fact that $G$ is a bipartite graph.
\begin{theorem}[\cite{KMN11,Kav16}]
  \label{thm:witness}
  A matching $M$ in $G = (A \cup B, E), \cal P$ is popular if and only if there exists a vector $\vec{\alpha} \in \{0, \pm 1\}^n$ such that
  $\sum_{u \in A \cup B} \alpha_u = 0$,
  \[ \alpha_a + \alpha_b \ \ \ge \ \ \wt_M(a,b)\ \ \ \forall\, (a,b)\in E\ \ \ \ \ \ \text{and}\ \ \ \ \ \ \alpha_u \ \ \ge\ \ \wt_M(u,u)\ \ \ \forall\, u \in A \cup B.\] 
\end{theorem}

For any popular matching $M$, a vector $\vec{\alpha}$ as given in Theorem~\ref{thm:witness} will be called $M$'s {\em witness} or {\em dual certificate}. A popular matching may have several witnesses. Any stable matching has $\vec{0}$ as a witness.

Call an edge $e$ {\em popular} if there is a popular matching in $G$ that contains $e$. For any popular edge $(a,b)$, it was shown in~\cite{FK20} (using complementary slackness) that $\alpha_a + \alpha_b = \wt_M(a,b)$. This will be a useful fact for us.
Another useful fact from \cite{FK20} (again by complementary slackness) is that if vertex $u$ is left unmatched in some popular matching then $\alpha_u = \wt_M(u,u)$ for every popular matching $M$; thus
$\alpha_u = 0$ if $u$ is left unmatched in $M$, else $\alpha_u = -1$.

\medskip

We now illustrate each of the above four tools on our example instance $G\new{, \mathcal{P}}$ from Fig.~\ref{fig:example}. 

\begin{enumerate}[1)]
    \item Recall the matching $M_1 = \{(a_1,b_1),(a_2,b_2)\}$. We will test if $M_1$ is popular/dominant using Theorems~\ref{thm:char-popular} and~\ref{thm:dominant}. First we label each edge $(a,b)$ not in $M_1$ by $(\vote_a(b,M_1(a)),\vote_b(a,M_1(b)))$, see Fig.\ref{fig:edge_labels}. Since no edge is labeled $(-,-)$, the subgraph $G_{M_1} = G$. Observe that $G_{M_1}$ has no forbidden alternating path/cycle as given in Theorem~\ref{thm:char-popular}. Thus $M_1$ is popular. However there is an $M_1$-augmenting path $\langle b_3, a_1, b_1, a_3\rangle$ in $G_{M_1}$,  so $M_1$ is not dominant (by Theorem~\ref{thm:dominant}).
\begin{figure}[h]
\tikzstyle{vertex} = [circle, draw=black, fill=black, inner sep=0pt,  minimum size=5pt]
\tikzstyle{edgelabel} = [circle, fill=white, inner sep=0pt,  minimum size=15pt]
\centering
	\pgfmathsetmacro{\d}{2}
\begin{tikzpicture}[scale=1, transform shape]
	\node[vertex, label=above:$a_1$] (a1) at (0,0) {};
	\node[vertex, label=left:$b_2$] (b2) at ($(a1) + (0, -\d)$) {};
	\node[vertex, label=above:$b_1$] (b1) at (\d,0) {};
	\node[vertex, label=right:$a_2$] (a2) at ($(b1) + (0, -\d)$) {};
	\node[vertex, label=above:$b_3$] (b3) at (-\d,0) {};
	\node[vertex, label=above:$a_3$] (a3) at ($(a1) + (2*\d, 0)$) {};
    
	\draw [very thick] (a1) -- node[edgelabel, near start] {2} node[edgelabel, above, rotate=90] {$(+,-)$} node[edgelabel, near end] {1} (b2);
	\draw [very thick, red, dotted] (a1) -- node[edgelabel, near start] {1} node[edgelabel, near end] {1} (b1);
    \draw [very thick] (a1) -- node[edgelabel, near start]  {3} node[edgelabel, above] {$(+,-)$} (b3);
	\draw [very thick, red, dotted] (a2) -- node[edgelabel, near start] {2} node[edgelabel, near end] {2} (b2);
	\draw [very thick] (a2) -- node[edgelabel, near start] {1} node[edgelabel, below, rotate=90] {$(+,-)$} node[edgelabel, near end] {2} (b1);
    \draw [very thick] (a3) -- node[edgelabel, near end] {3} node[edgelabel, above] {$(-,+)$}(b1);
\end{tikzpicture}
\caption{Edge labels with respect to $M_1 =\left\{(a_1,b_1),(a_2,b_2) \right\}$.}
\label{fig:edge_labels}
\end{figure}
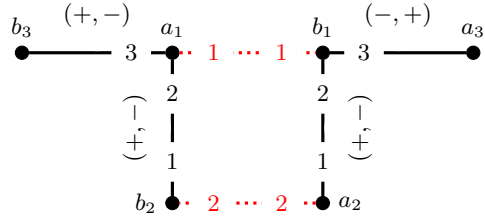

\item The bidirected graph $G'$ corresponding to $G$ is given in Fig.~\ref{fig:G'}. Observe that in \new{the transformed set of} preference lists, each vertex ranks its outgoing edges in their original order, followed by the incoming copies of the same edges in their original order. The only stable matching in $G'$ is $\{(a_1^+,b_2^-), (a_2^-,b_1^+)\}$ (marked by dashed teal edges in Fig.~\ref{fig:G'}) and it projects to 
$M_2 = \{(a_1,b_2),(a_2,b_1)\}$---this is the only dominant matching in~$G$.

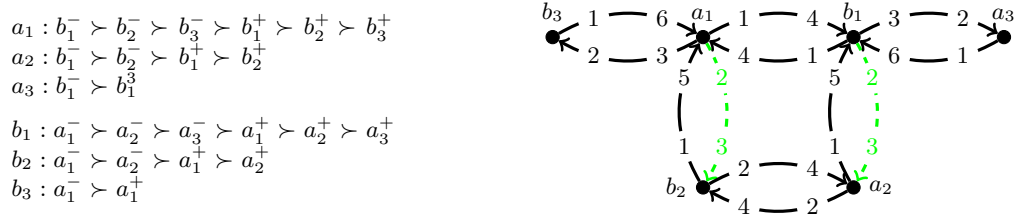
\begin{figure}[h]
\tikzstyle{vertex} = [circle, draw=black, fill=black, inner sep=0pt,  minimum size=5pt]
\tikzstyle{edgelabel} = [circle, fill=white, inner sep=0pt,  minimum size=12pt]
\centering
	\pgfmathsetmacro{\d}{2}
	\begin{minipage}{0.55\textwidth}
\[
\begin{array}{llllll}
a_1:  b_1^- \succ & b_2^- \succ & b_3^- \succ & b_1^+ \succ & b_2^+ \succ & b_3^+ \\
a_2:  b_1^- \succ & b_2^- \succ\ &   b_1^+ \succ & b_2^+    \ && \\ \vspace{2mm} 
a_3: b_1^- \succ &  b_1^+   &&&& \\  
b_1: a_1^- \succ & a_2^- \succ & a_3^- \succ & a_1^+ \succ & a_2^+ \succ & a_3^+\\
b_2: a_1^- \succ & a_2^- \succ &  a_1^+ \succ & a_2^+ && \\
b_3: a_1^- \succ & a_1^+ &&&&
\end{array}
\]
\end{minipage}\hspace{2mm}\begin{minipage}{0.4\textwidth}
\begin{tikzpicture}[scale=1.0, transform shape]
	\node[vertex, label=above:$a_1$] (a1) at (0,0) {};
	\node[vertex, label=left:$b_2$] (b2) at ($(a1) + (0, -\d)$) {};
	\node[vertex, label=above:$b_1$] (b1) at (\d,0) {};
	\node[vertex, label=right:$a_2$] (a2) at ($(b1) + (0, -\d)$) {};
	\node[vertex, label=above:$b_3$] (b3) at (-\d,0) {};
	\node[vertex, label=above:$a_3$] (a3) at ($(a1) + (2*\d, 0)$) {};
    
	\draw [->, very thick, teal, dashed] (a1)  to [out=-60,in=60] node[edgelabel, near start] {2} node[edgelabel, near end] {3} (b2);
	\draw [->, very thick, teal, dashed] (b1)  to [out=-60,in=60] node[edgelabel, near start] {2} node[edgelabel, near end] {3} (a2);
	\draw [<-, very thick] (a1)  to [out=-120,in=120] node[edgelabel, near start] {5} node[edgelabel, near end] {1} (b2);
	\draw [<-, very thick] (b1)  to [out=-120,in=120] node[edgelabel, near start] {5} node[edgelabel, near end] {1} (a2);
	\draw [->, very thick] (a1)  to [out=30,in=150] node[edgelabel, near start] {1} node[edgelabel, near end] {4} (b1);
	\draw [->, very thick] (b1)  to [out=30,in=150] node[edgelabel, near start] {3} node[edgelabel, near end] {2} (a3);
	\draw [->, very thick] (b3)  to [out=30,in=150] node[edgelabel, near start] {1} node[edgelabel, near end] {6} (a1);
	\draw [->, very thick] (b2)  to [out=30,in=150] node[edgelabel, near start] {2} node[edgelabel, near end] {4} (a2);
	\draw [->, very thick] (a1)  to [out=-150,in=-30] node[edgelabel, near start] {3} node[edgelabel, near end] {2} (b3);
	\draw [->, very thick] (b1)  to [out=-150,in=-30] node[edgelabel, near start] {1} node[edgelabel, near end] {4} (a1);
	\draw [->, very thick] (a3)  to [out=-150,in=-30] node[edgelabel, near start] {1} node[edgelabel, near end] {6} (b1);
	\draw [->, very thick] (a2)  to [out=-150,in=-30] node[edgelabel, near start] {2} node[edgelabel, near end] {4} (b2);
	
\end{tikzpicture}
\end{minipage}
\caption{The bidirected graph $G'$ corresponding to $G$, and the \new{transformed} lists.}
\label{fig:G'}
\end{figure}

\item The set of matched vertices is the same for all popular matchings in the above instance: $V(M_1) = V(M_2) = \{(a_1,a_2,b_1,b_2\}$. 
\item We will construct a witness $\vec{\alpha}$ for $M_2$, see Fig.~\ref{fig:witness}. So $\vec{\alpha} \in \{0, \pm 1\}^6$. Since $M_2$ leaves $a_3$ and $b_3$ unmatched, it has to be the case that $\alpha_{a_3} = \alpha_{b_3} = 0$. We also know that $\alpha_{a_1} = \alpha_{b_1} = 1$ because $\alpha_{a_1} + \alpha_{b_1} \ge \wt_{M_2}(a_1,b_1) = 2$. Since $\sum_{u \in A \cup B} \alpha_u = 0$, the only possibility for the remaining two vertices is $\alpha_{a_2} = \alpha_{b_2} = -1$. We have $\alpha_u \ge \wt_{M_2}(u,u)$ for all vertices $u$ since
$\wt_{M_2}(a_3,a_3) = \wt_{M_2}(b_3,b_3) = 0$ and $\wt_{M_2}(v,v) = -1$ for other vertices $v$.

Observe that $\alpha_{a_1} + \alpha_{b_3} = 1 > 0 = \wt_{M_2}(a_1,b_3)$ and
$\alpha_{a_3} + \alpha_{b_1} = 1 > 0 = \wt_{M_2}(a_3,b_1)$. For the remaining 4 edges, the corresponding constraint is tight. That is,
$\alpha_a + \alpha_b = \wt_{M_2}(a,b)$ for $a \in \{a_1,a_2\}$ and $b \in \{b_1,b_2\}$.

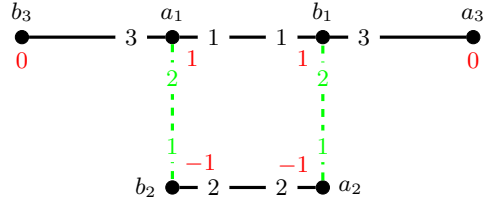
\begin{figure}[h]
\tikzstyle{vertex} = [circle, draw=black, fill=black, inner sep=0pt,  minimum size=5pt]
\tikzstyle{edgelabel} = [circle, fill=white, inner sep=0pt,  minimum size=12pt]
\centering
	\pgfmathsetmacro{\d}{2}
\begin{tikzpicture}[scale=1, transform shape]
	\node[vertex, label=above:$a_1$, label=-45:{\color{red}$1$}] (a1) at (0,0) {};
	\node[vertex, label=left:$b_2$,label=60:{\color{red}$-1$}] (b2) at ($(a1) + (0, -\d)$) {};
	\node[vertex, label=above:$b_1$,label=-135:{\color{red}$1$}] (b1) at (\d,0) {};
	\node[vertex, label=right:$a_2$,label=135:{\color{red}$-1$}] (a2) at ($(b1) + (0, -\d)$) {};
	\node[vertex, label=above:$b_3$, label=below:{\color{red}$0$}] (b3) at (-\d,0) {};
	\node[vertex, label=above:$a_3$, label=below:{\color{red}$0$}] (a3) at ($(a1) + (2*\d, 0)$) {};
    
	\draw [very thick, teal, dashed] (a1) -- node[edgelabel, near start] {2} node[edgelabel, near end] {1} (b2);
	\draw [very thick] (a1) -- node[edgelabel, near start] {1} node[edgelabel, near end] {1} (b1);
    \draw [very thick] (a1) -- node[edgelabel, near start] {3} (b3);
	\draw [very thick] (a2) -- node[edgelabel, near start] {2} node[edgelabel, near end] {2} (b2);
	\draw [very thick, teal, dashed] (a2) -- node[edgelabel, near start] {1} node[edgelabel, near end] {2} (b1);
    \draw [very thick] (a3) -- node[edgelabel, near end] {3} (b1);
\end{tikzpicture}
\caption{A witness $\vec{\alpha}$ constructed for $M_2 =\left\{(a_1,b_2),(a_2,b_1) \right\}$ is denoted by the red labels next to the vertices.}
\label{fig:witness}
\end{figure}
\end{enumerate}

\section{Finding an unstable popular matching}
\label{sec:dom-vs-stab}
We are given $G = (A \cup B,E)$ with strict preference lists and the problem is to decide if every popular matching in $G$ is also stable, i.e., if $\{$popular matchings$\} = \{$stable matchings$\}$ or not in $G$. In this section we show an efficient algorithm to answer this question.

\begin{pr}
    \inp A bipartite graph $G = (A \cup B,E)$ with strict preference lists.\\
	\ques 
	If there is an unstable popular matching in $G$.
\end{pr}

Let $G$ admit an unstable popular matching~$M$ and let $\vec{\alpha}\in\{0, \pm 1\}^n$ be $M$'s witness. Let $A_0$ be the set of vertices $a\in A$ with $\alpha_a = 0$ and let $B_0$ be the set of vertices $b\in B$ with $\alpha_b = 0$.

\begin{figure}[ht]
	\vspace*{-5mm}
\centerline{\resizebox{0.6\textwidth}{!}{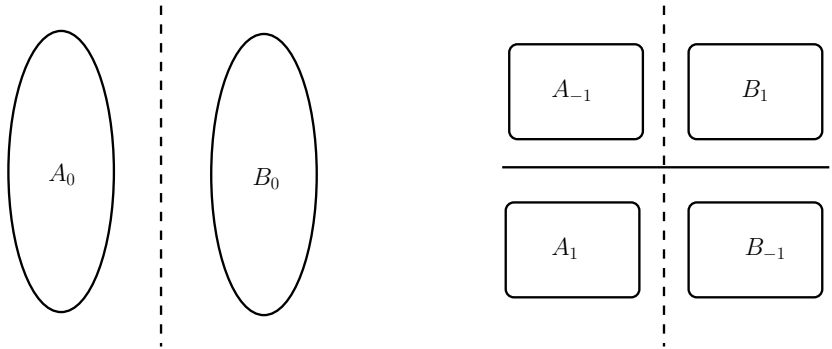}}
\caption{$M_0$ is the matching $M$ restricted to vertices in $A_0\cup B_0$ and $M_1$ is the matching $M$ restricted to remaining vertices.}
\label{fig:first}
\vspace*{-5mm}
\end{figure}

Let $M_0$ be the set of edges $(a,b) \in M$ such that $\alpha_a = \alpha_b = 0$. Let $M_1$ be the set of edges $(a,b) \in M$ such that $\alpha_a,\alpha_b \in \{\pm 1\}$. Note that $M = M_0 \cupdot M_1$, since the parities of $\alpha_a$ and $\alpha_b$ have to be the same for any popular edge due to the tightness of the constraint 
$\alpha_a + \alpha_b = \wt_M(a,b) = 0$. 

The construction can be followed in Fig.~\ref{fig:first}. $A \setminus A_0$ has been further split into $A_1 \cup A_{-1}$: $A_1$ is the set of those vertices $a$ with $\alpha_a = 1$ and $A_{-1}$ is the set of those vertices $a$ with $\alpha_a = -1$; similarly, $B \setminus B_0$ has been further split into $B_1 \cup B_{-1}$: $B_1$ is the set of those vertices $b$ with $\alpha_b = 1$ and $B_{-1}$ is the set of those vertices $b$ with $\alpha_b = -1$. We have $M_1 \subseteq (A_1\times B_{-1})\cup(A_{-1}\times B_1)$ since for any edge $(a,b)$ in $M$, we have $\alpha_a + \alpha_b = \wt_M(a,b) = 0$.

Now we run a transformation $M_0 \leadsto D$ as given in~\cite{CK18}: let $G_0$ be the graph $G$ restricted to $A_0 \cup B_0$. Run Gale-Shapley algorithm in the graph $G'_0$ with the starting matching $M'_0 = \{(u^+,v^-): (u,v) \in M_0 \ \text{and}\ u\in A_0, v\in B_0\}$, where $G'_0$ is the graph obtained from $G_0$ as described in Section~\ref{prelims}.
That is, instead of starting with the empty matching, we start with the matching $M'_0$ in $G'_0$;
unmatched men in $A_0$ propose in decreasing order of preference and whenever a woman receives a proposal from a neighbor that she prefers to her current partner (her preferences as given in $G'_0$), she rejects her current partner and accepts this proposal. This results in a stable matching in $G'_0$, equivalently, a dominant matching $D$ in $G_0$, see Section~\ref{prelims}. 
It was moreover shown in~\cite{CK18} that $M^* = M_1 \cup D$ is a dominant matching in $G$. 
 We include a new and much simpler proof of this fact below. 

\begin{new-claim}
\label{clm:dominant-matching}
The matching $M^* = M_1 \cup D$ is dominant in $G$.
\end{new-claim}
\begin{proof}
The dominant matching $D$ was obtained as the linear image of a stable matching (call it $D'$) in $G'_0$.
Let $A'_1$ be the set of $a \in A_0$ such that $(a^+,b^-) \in D'$ for some $b \in B_0$ and $A'_{-1} = A_0 \setminus A'_1$.
Similarly, let $B'_1$ be the set of $b \in B_0$ such that $(a^-,b^+) \in D'$ for some $a \in A_0$ and $B'_{-1} = B_0 \setminus B'_1$. See Fig.~\ref{fig:Appendix1}.

\begin{figure}[ht]
	\vspace*{-5mm}
\centerline{\resizebox{0.72\textwidth}{!}{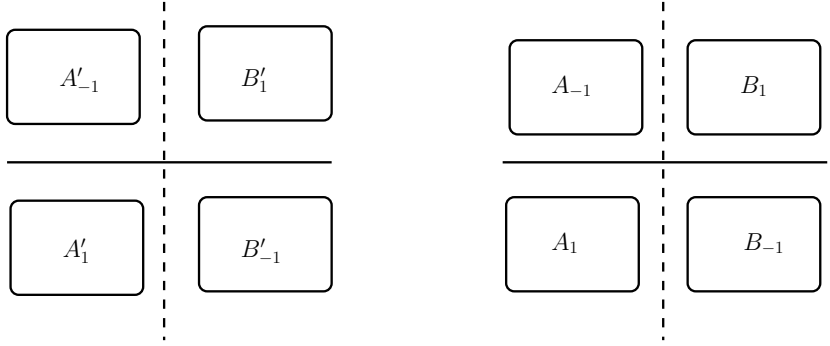}}
\caption{We transformed the stable matching $M_0$ on $A_0 \cup B_0$ to the dominant matching $D$: 
this partitions $A_0$ into $A'_1 \cup A'_{-1}$ and $B_0$ into $B'_1 \cup B'_{-1}$. We also have $A \setminus A_0 = A_1 \cup A_{-1}$ and 
$B \setminus B_0 = B_1 \cup B_{-1}$.}
\label{fig:Appendix1}
\vspace*{-5mm}
\end{figure}

Observe that every vertex in $A'_1 \cup B'_1$ is matched in $D$, however not every vertex in $A'_{-1} \cup B'_{-1}$ is necessarily
matched in $D$. That is, $A'_{-1} \cup B'_{-1}$ may contain unmatched vertices.

The popular matching $M$ has a witness $\vec{\alpha} \in \{0, \pm 1\}^n$: recall that this was the witness used to partition $A \cup B$ into $A_0 \cup B_0$ and $(A \setminus A_0) \cup (B \setminus B_0)$. We have $\alpha_u = 0$ for all $u \in A_0 \cup B_0$ and 
$\alpha_u \in \{\pm 1\}$ for all $u \in (A \setminus A_0) \cup (B \setminus B_0)$.

In order to prove the popularity of  $M^* = M_1 \cup D$, we will show a witness $\vec{\beta}$ as follows. 
For the vertices outside $A_0\cup B_0$, set $\beta_u = \alpha_u$ since nothing has changed for these vertices in 
the transformation $M \leadsto M^*$. For the vertices in $A_0\cup B_0$, set $\beta_u = 1$ for $u \in A'_1 \cup B'_1$ and 
$\beta_u = -1$ for every matched vertex $u \in A'_{-1} \cup B'_{-1}$. For every unmatched vertex $u$, set $\beta_u = 0$.

Thus $\beta_u \ge \wt_{M^*}(u,u)$ for all $u$ and $\sum_{u:A\cup B} \beta_u = 0$: this is because $D \subseteq (A'_1\times B'_{-1})\cup(A'_{-1}\times B'_1)$, so for every edge $(a,b) \in D$, we have $\beta_a + \beta_b = 0$;
recall that $\alpha_a + \alpha_b = 0$ for all $(a,b) \in M_1$.
What is left to show is that $\beta_a + \beta_b \ge \wt_{M^*}(a,b)$ for every edge $(a,b)$.

The correctness of the Gale-Shapley algorithm in $G'_0$ to compute popular matchings in $G_0$ immediately implies that every edge with both endpoints in $A_0 \cup B_0$ is covered by the sum of $\beta$-values of its endpoints (see \cite{FKPZ18}). We will now show that every edge with one endpoint in $A_0\cup B_0$ and the other endpoint in $(A \setminus A_0)\cup (B \setminus B_0)$ is also covered. 
\begin{itemize}
\item Every edge in $A'_1\times B_1$ is covered since $\beta_a + \beta_b = 2 \ge \wt_{M^*}(a,b)$.
Similarly every edge in $A_1\times B'_1$ is also covered.

\item Consider any edge $(a,b) \in A'_1\times B_{-1}$. We have $\alpha_a = 0$ and $\alpha_b = -1$ and so $\wt_M(a,b) \le -1$, 
i.e., $\wt_M(a,b) = -2$ since this is a value in $\{0, \pm 2\}$. This means $b$ prefers $M(b) = M^*(b)$ to $a$. Thus $\wt_{M^*}(a,b) \le 0$.
Since $\beta_a = 1$ and $\beta_b = -1$, we have $\beta_a + \beta_b \ge \wt_{M^*}(a,b)$. 
Similarly every edge in $A_{-1} \times B'_1$ is also covered.

\item Consider any edge $(a,b)$ in  $A_{-1}\times B'_{-1}$. We have $\wt_M(a,b) \le -1$ which means that 
$\wt_M(a,b) = -2$. That is, both $a$ and $b$ prefer their partners in $M$ to each other. We will now show that 
$\wt_{M^*}(a,b) = -2$. Since $M(a) = M^*(a)$, nothing has changed for $a$ and so $a$ prefers $M^*(a)$ to $b$.
We claim that $M^*(b) \succeq_b M(b)$, i.e., $b$ is no worse in $M^*$ than in $M$.

This is because we ran Gale-Shapley algorithm in $G'_0$ with $M'_0$ as the starting matching. 
So if $b \in B'_{-1}$ changed its partner from $u^+$ in $M'_0$ to $v^+$ in $D'$ then $b$ prefers $v$ to $u$.
Thus every $b \in B'_{-1}$ has at least as good a partner in $D$ as in $M_0$. Hence
$\wt_{M^*}(a,b) = -2 = \beta_a + \beta_b$.
By the same reasoning, we can argue that every edge in  $A_1\times B'_{-1}$ is also covered.

\item Consider any edge $(a,b)$ in $A'_{-1}\times B_{-1}$. We have $\alpha_a = 0$
and $\alpha_b = -1$ and so $\wt_M(a,b) \le -1$, 
i.e., $\wt_M(a,b) = -2$. So $b$ prefers $M(b) = M^*(b)$ to $a$. We claim that $M^*(a) \succeq_a M(a)$. This will imply that
$\wt_{M^*}(a,b) = -2$.

Recall that $D'$ is a stable matching in $G'_0$: here every vertex prefers outgoing edges to incoming edges, 
thus every $a \in A'_{-1}$ gets at least as good a partner as $S^*(a)$ in $D$, where $S^*$ is the men-optimal stable matching in $G_0$; 
in turn, $S^*(a) \succeq_a M_0(a)$ for all $a \in A$.
Hence $M^*(a) \succeq_a M(a)$ for every $a \in A'_{-1}$. Thus $\wt_{M^*}(a,b) = -2 = \beta_a + \beta_b$. We can similarly show that every edge in $A'_{-1}\times B_1$ is also covered.
\end{itemize}

Thus $M^*$ is a popular matching in $G$. We will now show that $M^*$ is a {\em dominant} matching in~$G$. Observe that 
$\beta_u \in \{\pm 1\}$ for every matched vertex $u$: we will use this fact to show that $M^*$ is dominant. 
Recall that $\beta_u = 0$ for all unmatched vertices $u$.
Let $\rho = \langle a_0,b_1,a_1,\ldots,b_k,a_k,b_{k+1}\rangle$ be any $M^*$-augmenting path in $G$. 
We have $\beta_{a_0} = \beta_{b_{k+1}} = 0$, hence $\beta_{b_1}= \beta_{a_k} = 1$, and so $\beta_{a_1} = \beta_{b_k} = -1$. 

Either (i)~$\beta_{b_2} = -1$ or (ii)~$\beta_{b_2} = 1$ which implies that $\beta_{a_2} = -1$. It is now easy to see that in the path $\rho$,
for some $i \in \{1,\ldots,k-1\}$ the edge $(a_i,b_{i+1})$ has to be labeled $(-,-)$. That is, $\rho$ is {\em not} an
augmenting path in $G_{M^*}$. Thus there is no $M^*$-augmenting path in $G_{M^*}$, hence $M^*$ is a dominant matching in $G$
(by Theorem~\ref{thm:dominant}). \qed
\end{proof}

Since $M$ is an unstable matching, there is an edge $(a,b)$ that blocks~$M$. Since $\wt_M(a,b) = 2$, the endpoints of a blocking edge $(a,b)$ have to satisfy $\alpha_a = \alpha_b = 1$; so $a \in A_1$ and $b \in B_1$. The edge $(a,b)$ blocks $M^*$ as well since the matching $M_1$ was unchanged by this transformation of $M_0$ to $D$, so $M^*(a) = M_1(a)$ and $M^*(b) = M_1(b)$, thus $a$ and $b$  prefer each other to their respective partners in~$M^*$. So $M^*$ is an unstable dominant matching and Lemma~\ref{non-stab-domn} follows.

\begin{lemma}
\label{non-stab-domn}
If $G,{\cal P}$ has an unstable popular matching then $G,{\cal P}$ admits an unstable dominant matching.
\end{lemma}

Hence in order to answer the question of whether every popular matching in $G$ is stable or not, we need to decide if there exists a dominant matching $M$ in $G$ with a blocking edge. 
We present a simple combinatorial algorithm for this problem.

\medskip

Our algorithm is based on the equivalence between dominant matchings in $G$ and stable matchings in $G'$.
Our task is to determine if there exists a stable matching in $G'$ that includes a pair of edges $(a^+,v^-)$ and $(u^-,b^+)$ such that $a$ and $b$ prefer each other to $v$ and $u$, respectively, in~$G$. It is easy to decide in $O(m^3)$ time whether such a stable matching exists or not in~$G'$. 
\begin{itemize}
\item For every pair of edges $e_1 = (a,v)$ and $e_2 = (u,b)$ in $G$ such that
$a$ and $b$ prefer each other to $v$ and $u$, respectively: determine if there is a stable matching in $G'$ that contains the pair of edges
$(a^+,v^-)$ and~$(u^-,b^+)$. 
\end{itemize}

In the graph $G'$, we modify Gale-Shapley algorithm so that $b$ rejects proposals 
from all neighbors ranked worse than $u^-$ and $v$ rejects all proposals from neighbors ranked worse than~$a^+$. 
If the resulting algorithm returns a stable matching that contains the edges $(a^+,v^-)$ and~$(u^-,b^+)$, then we have the desired matching; else $G'$ 
has no stable matching that contains this particular pair of edges.

In order to determine if there exists an unstable dominant matching, we may need to go through all
pairs of edges $(e_1,e_2) \in E\times E$. Since we can determine in linear time if there exists a stable matching in $G'$ with any particular pair of edges~\cite{GI89},
the running time of this algorithm is $O(m^3)$.

\paragraph{A faster algorithm.} It is easy to improve the running time to $O(m^2)$. For each $(a,b) \in E$ we check the following.
\begin{itemize}
\item[($\circ$)] Does there exist a stable matching in $G'$ such that
\begin{inparaenum}[(1)]
\item $a^+$ is matched to a neighbor that is ranked worse than $b^-$ in $a$'s list, and
\item $b^+$ is matched to a neighbor that is ranked worse than $a^-$ in $b$'s list?
\end{inparaenum} 
\end{itemize}

We modify the Gale-Shapley algorithm in $G'$ so that (1)~$b$ rejects all offers from superscript~$+$ neighbors, i.e., $b$ accepts proposals only from superscript~$-$ neighbors, and (2)~every neighbor of $a$ that is ranked better than $b^-$ rejects proposals from~$a^+$. 

Suppose ($\circ$) holds.
Then this modified Gale-Shapley algorithm returns among all such stable matchings, the most men-optimal and women-pessimal one~\cite{GI89}. 
Thus among all stable matchings that match $a^+$ to a neighbor ranked worse than $b^-$ and that include some edge $(\ast,b^+)$, 
the matching returned by the above algorithm matches $b$ to its least preferred neighbor and $a$ to its most preferred neighbor.

Hence if the modified Gale-Shapley algorithm returns a matching that is \begin{inparaenum}[(i)] \item unstable or 
\item includes an edge $(a^-,\ast)$ or
\item matches $b^+$ to a neighbor better than $a^-$, then there is no dominant matching $M$ in $G$ such that the pair $(a,b)$ blocks~$M$.
\end{inparaenum}
Else we have the desired stable matching in $G'$, call this matching~$M'$. 

The projection of the matching $M'$ on to the edge set of $G$ will be a dominant matching in $G$ with $(a,b)$ as a blocking edge.
Since we may need to go through all edges in $E$ and the time taken for any edge $(a,b)$ is $O(m)$, the entire running time of this
algorithm is~$O(m^2)$. We have thus shown the following theorem.

\begin{theorem}
\label{thm:unstable}
Given $G = (A \cup B,E)$ on $m$ edges and strict preference lists of its vertices, we can decide in $O(m^2)$ time whether every popular matching in $G$ is stable or not; if not, we can return an unstable popular matching.
\end{theorem}

\section{Finding a non-dominant popular matching}
\label{sec:stable}

Given an instance $G = (A \cup B, E)\new{, \mathcal{P}}$, the problem we consider here is to decide if every popular matching is also
dominant, i.e., to decide if $\{$popular matchings$\} = \{$dominant matchings$\}$ or not in~$G$.

\begin{pr}
    \inp A bipartite graph $G = (A \cup B,E)$ with strict preference lists.\\
	\ques If there is a non-dominant popular matching in $G$.
\end{pr}

In this section, we show the following. 
\begin{theorem}
\label{thm:non-dominant}
Given $G = (A \cup B,E)$ with strict preference lists, it is NP-complete to decide if $G$ admits a popular matching that is not dominant.
\end{theorem}

We start with the following lemma, that is the counterpart of Lemma~\ref{non-stab-domn}.

\begin{lemma}
\label{non-domn-stable}
If $G,{\cal P}$ has a non-dominant popular matching $M$ then $G,{\cal P}$ admits a non-dominant stable matching $N$. If $M$ is given, then $N$ can be found efficiently.
\end{lemma}
\begin{proof}
  Let $M$ be a non-dominant popular matching in $G$ and $\vec{\alpha}\in\{0, \pm 1\}^n$ its witness (see Theorem~\ref{thm:witness}). 
  We will use the decomposition illustrated in Fig.~\ref{fig:first} to show the existence of a non-dominant stable matching in $G$. As per the decomposition in Fig.~\ref{fig:first}, $M = M_0 \cupdot M_1$. Since $M$ is not dominant, there exists an $M$-augmenting path $\rho$ in $G_M$ (by Theorem~\ref{thm:dominant}). The endpoints of $\rho$ (call them $u$ and~$v$) are unmatched in $M$, hence $\alpha_u = \alpha_v = 0$ (see Section~\ref{prelims})  
  and so $u$ and $v$ are in $A_0 \cup B_0$.

  In the graph $G_M$, the vertices in $B_{-1} \cup A_{-1}$ (these vertices have $\alpha$-value $-1$) are adjacent  only to vertices in $A_1 \cup B_1$ as all other edges incident to vertices in $B_{-1} \cup A_{-1}$ are labeled $(-,-)$, and these are not present in~$G_M$. Suppose the $M$-alternating path $\rho$ leaves the vertex set $A_0 \cup B_0$, i.e., suppose it contains a non-matching edge between $A_0 \cup B_0$ and $A_1 \cup B_1$. Since the partners of vertices in  $A_1 \cup B_1$ are in $B_{-1} \cup A_{-1}$, the path $\rho$ can never return to  $A_0 \cup B_0$. However we know that the last vertex $v$ of $\rho$ is in $A_0\cup B_0$. Thus $\rho$ never leaves $A_0 \cup B_0$, i.e., $\rho$ is an  $M_0$-augmenting path in $G^*_{M_0}$, where $G^*$ is the graph $G$ restricted to vertices in $A_0 \cup B_0$.

  We now run a transformation $M_1 \leadsto S$ as given in \cite{CK18} to convert $M_1$ into a stable matching $S$ as follows. The matching $S$ is obtained by  running Gale-Shapley algorithm in the subgraph $G_1$, which is the graph $G$ restricted to $(A\setminus A_0)\cup(B\setminus B_0)$: however, rather than starting with the empty matching, we start with the matching given by $M_1 \cap (A_{-1}\times B_1)$. So unmatched men (these are vertices in $A_1$ to begin with) propose in decreasing order of preference and whenever a woman receives a proposal from a neighbor that she prefers to her current partner (her preferences as given in $G$), she rejects her current partner and accepts this proposal. This results in a stable matching $S$ in $G_1$.

  It was shown in \cite{CK18} that $N = M_0 \cup S$ is a stable matching in $G$. We include a new and simple proof of this below (see Claim~\ref{clm:stable-matching}). Since  $\rho$ is an $M_0$-augmenting path in $G^*_{M_0}$ and $M_0$ is a subset of $N$, it follows that $\rho$ is an $N$-augmenting path in $G_N$. Thus $N$ is a non-dominant stable matching in~$G$. \qed
\end{proof}

\begin{new-claim}
\label{clm:stable-matching}
The matching $N = M_0 \cup S$ is a stable matching in $G$.
\end{new-claim}
\begin{proof}
The matching $S$ is obtained by running Gale-Shapley algorithm in the graph $G_1$ on vertex set $(A \setminus A_0)\cup(B \setminus B_0)$.
We did not compute the matching $S$ from scratch --- we started with edges of the matching $M_1$ restricted to $A_{-1} \times B_1$. So in the resulting matching $S$, it is easy to see the following two useful properties:
\begin{itemize}
\item $S(b) \succeq_b M_1(b)$ for every $b \in B_1$. This is because to begin with, every $b \in B_1$ is matched to $M_1(b)$ and $b$ will 
change her partner only if she receives a proposal from a neighbor better than $M_1(b)$. 
\item $S(a) \succeq_a M_1(a)$ for every $a \in A_1$. This is because all vertices in $B_{-1}$ are unmatched in our starting matching and 
every $b \in B_{-1}$ prefers her partner in $M_1$ to any neighbor in $A_{-1}$ (since every edge in $A_{-1}\times B_{-1}$ is labeled $(-,-)$ with respect to $M_1$).
Thus in the matching $S$, $a$ will get accepted either by $M_1(a)$ or a better neighbor.
\end{itemize}

It is now easy to show that $N$ is a stable matching. We already know that $M_0$ is a stable matching on $A_0 \cup B_0$ and $S$ is a stable
matching on $(A \setminus A_0) \cup (B \setminus B_0)$. It is left to show that no edge $(a,b)$ with one endpoint in $A_0 \cup B_0$ and
another endpoint in $(A \setminus A_0) \cup (B \setminus B_0)$ {\em blocks} $N$, i.e., we need to show that $\wt_N(a,b) \le 0$ for every such edge $(a,b)$.

Suppose $a \in A_0$ and $b \in B_{-1}$. Let $\vec{\alpha}$ be the witness of $M$ used to partition $A\cup B$ into $A_0\cup B_0$ and $(A \setminus A_0) \cup (B\setminus B_0)$. So $\alpha_a = 0$ and $\alpha_b = -1$, hence
$\wt_M(a,b) \le -1$, i.e., $\wt_M(a,b) = -2$. Thus $a$ prefers $M_0(a) = N(a)$ to $b$. Hence $\wt_N(a,b) \le 0$. We can similarly show that
$\wt_N(a,b) \le 0$ for $a \in A_{-1}$ and $b \in B_0$.

Suppose $a \in A_0$ and $b \in B_1$. Then $\alpha_a = 0$ and $\alpha_b = 1$ and so $\wt_M(a,b) \le 1$, i.e., $\wt_M(a,b) \le 0$. 
So if $a$ prefers $M_0(a) = N(a)$ to $b$ then we can immediately conclude that $\wt_N(a,b) \le 0$. Else $b$ prefers $M_1(b)$ to $a$ and
we have noted above that $S(b) \succeq_b M_1(b)$ for every $b \in B_1$. Thus $b$ prefers $N(b) = S(b)$ to $a$ and so
$\wt_N(a,b) \le 0$. We can similarly argue that $\wt_N(a,b) \le 0$ for every $a \in A_1$ and $b \in B_0$. 
This finishes the proof of this claim. \qed
\end{proof}

	Our problem now is to decide if there exists a non-dominant stable matching $N$ in $G$, i.e., a stable matching $N$ with an $N$-augmenting path in $G_N$ (see Theorem~\ref{thm:dominant}). We show a reduction from 3SAT. Given a 3SAT formula $\phi$, we transform $\phi = C_1 \wedge \cdots \wedge C_m$ as follows: let $X_1,\ldots,X_n$ be the $n$ variables in $\phi$. For each $i$:
	\begin{itemize}
		\item replace all occurrences of $\neg X_i$ in $\phi$ with $X_{n+i}$ (a single new variable);
		\item add the clauses $X_i \vee X_{n+i}$ and $\neg X_i \vee \neg X_{n+i}$ to capture $\neg X_i \iff X_{n+i}$.
	\end{itemize}
	
	Thus, the updated formula is $\phi = C_1 \wedge \cdots \wedge C_m \wedge C_{m+1} \wedge\cdots\wedge C_{m+n} \wedge D_{m+n+1} \wedge \cdots \wedge D_{m+2n}$, where $C_1,\ldots,C_m$ are the original $m$ clauses with negated literals substituted by new variables and for $1 \le i \le n$: $C_{m+i}$ is the clause $X_i \vee X_{n+i}$ and $D_{m+n+i}$ is the clause $\neg X_i \vee \neg X_{n+i}$. Corresponding to the above formula $\phi$, we will construct an instance $G\new{, \mathcal{P}}$ whose high-level picture is shown in Fig.~\ref{newfig1:example}. 
 
	\begin{figure}[ht]
		\centering
		\begin{tikzpicture}[scale=0.9, transform shape, very thick]
		\pgfmathsetmacro{\b}{1.5}
		\pgfmathsetmacro{\c}{45}
		\pgfmathsetmacro{\a}{30}
		
		\node[vvertex, label=below:$s$] (s) at (1*\b, 0) {};
		\node[uvertex, label=below:$u_0$] (u0) at (2*\b, 0) {};
		\node[vvertex, label=below:$v_0$] (v0) at (3*\b, 0) {};
		\node[uvertex, label=below:$u_1$] (u1) at (5*\b, 0) {};
		\node[vvertex, label=below:$v_1$] (v1) at (6*\b, 0) {};
		\node[uvertex, label=below:$u_{m+2n-1}$] (ulastbutone) at (7*\b, 0) {};
		\node[vvertex, label={below, yshift=-3mm}:$v_{m+2n-1}$] (vlastbutone) at (8*\b, 0) {};
		\node[uvertex, label={below, yshift=-3mm}:$u_{m+2n}$] (ulast) at (10*\b, 0) {};
		\node[vvertex, label=below:$v_{m+2n}$] (vlast) at (11*\b, 0) {};
		\node[uvertex, label=below:$t$] (t) at (12*\b, 0) {};
		
		\node[dot](p1) at (6.5*\b, 0) {};
		\node[dot](p2) at (6.4*\b, 0) {};
		\node[dot](p3) at (6.6*\b, 0) {};
		
		\draw (s) -- node[edgelabel, near start] {1} node[edgelabel, near end] {2} (u0);
		\draw [MyPurple] (u0) -- node[edgelabel, near start] {1} node[edgelabel, near end] {1} (v0);
		\draw [MyPurple] (u1) -- node[edgelabel, near start] {4} node[edgelabel, near end] {1} (v1);
		\draw [MyPurple] (ulastbutone) -- node[edgelabel, near start] {1} node[edgelabel, near end] {3} (vlastbutone);
		\draw [MyPurple] (ulast) -- node[edgelabel, near start] {1} node[edgelabel, near end] {1} (vlast);
		\draw (vlast) -- node[edgelabel, near start] {2} node[edgelabel, near end] {1} (t);
		
		\draw [] (v0) to[out=\c,in=180-\c] node[edgelabel, very near start] {2} node[edgebox] {$Z_{1,1}$} node[edgelabel, very near end] {1} (u1);
		\draw (v0) -- node[edgelabel, very near start] {3} node[edgebox] {$Z_{1,2}$} node[edgelabel, very near end] {2} (u1);
https://www.overleaf.com/project/5bc995a224f85304d9329f5e		\draw [] (v0) to[out=-\c,in=-180+\c] node[edgelabel, very near start] {4} node[edgebox] {$Z_{1,3}$} node[edgelabel, very near end] {3} (u1);
		
		\draw (vlastbutone) to[out=\a,in=180-\a] node[edgelabel, very near start] {1} node[edgebox] {$Z_{m+2n,1}$} node[edgelabel, very near end] {2} (ulast);
		\draw (vlastbutone) to[out=-\a,in=-180+\a] node[edgelabel, very near start] {2} node[edgebox] {$Z_{m+2n,2}$} node[edgelabel, very near end] {3} (ulast);
		
		\end{tikzpicture}
		\caption{The high-level picture of the instance $G\new{, \mathcal{P}}$.}
		\label{newfig1:example}
	\end{figure}
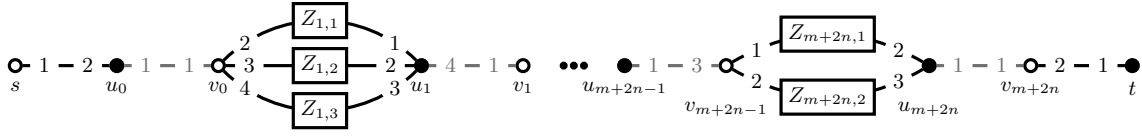
	
$G$ is the series composition of an edge $(s,u_0)$, a gadget for each clause $i$ that starts with node $v_{i-1}$ and ends with node $u_i$, and finally, an edge $(v_{m+2n},t)$. The gadget corresponding to clause $i$ contains the parallel composition of disjoint gadgets $Z_{i,j}$ for each literal $j$ in clause $i$. Note that  gadgets $Z_{i,j}$ associated with positive literals (so $i\leq m+n$) are different from those associated with 
 negated literals (here $i\geq m+n+1$).

For a variable $x$, we use $a_i,b_i,a'_i,b'_i$ to denote the 4 vertices in the gadget of $x$ in $C_i$ and $c_k,d_k,c'_k,d'_k$ to denote the 4 vertices in the gadget of $\neg x$ in $D_k$. The edge $(a_i,d_k)$ will be called a {\em consistency} edge. This connects the gadget of a positive literal to the (unique) corresponding negative literal (see Fig.~\ref{newfig2:example}). 
Definitions of preference lists are given in Section~\ref{sec:preferences}. Here, we explain the main steps of the proof. Define $F:=\{(u_i,v_i):0 \le i \le m+2n\}$ to be the set of \emph{basic edges}.

\begin{figure}[ht]
	\centering
		\begin{tikzpicture}[scale=0.9, transform shape, very thick]
		\pgfmathsetmacro{\b}{2.3}
		\node[vvertex, label=above:$v_{i-1}$] (vi-1) at (0, 0) {};
		\node[uvertex, label=above:$a_i$] (ai) at (1*\b, 0) {};
		\node[vvertex, label=above:$b_i$] (bi) at (2*\b, 0) {};
		\node[uvertex, label=above:$u_i$] (ui) at (3*\b, 0) {};
		\node[vvertex, label=above:$v_{k-1}$] (vk-1) at (4*\b, 0) {};
		\node[uvertex, label=above:$c_k$] (ck) at (5*\b, 0) {};
		\node[vvertex, label=above:$d_k$] (dk) at (6*\b, 0) {};
		\node[uvertex, label=above:$u_k$] (uk) at (7*\b, 0) {};
		\node[vvertex, label=below:$b'_i$] (b'i) at (1*\b, -\b) {};
		\node[uvertex, label=below:$a'_i$] (a'i) at (2*\b, -\b) {};
		\node[vvertex, label=below:$d'_k$] (d'k) at (5*\b, -\b) {};
		\node[uvertex, label=below:$c'_k$] (c'k) at (6*\b, -\b) {};
		
		\draw (vi-1) -- node[edgelabel,  near start] {2} node[edgelabel,  near end] {3} (ai);
		\draw (bi) -- node[edgelabel,  near start] {3} node[edgelabel,  near end] {1} (ui);
		\draw (vk-1) -- node[edgelabel,  near start] {1} node[edgelabel,  near end] {3} (ck);
		\draw (dk) -- node[edgelabel,  near start, scale=0.9] {$(\infty - 1)$} node[edgelabel,  near end] {2} (uk);		
		
		\draw [blue, dashed] (ai) -- node[edgelabel,  near start] {1} node[edgelabel,  near end] {2} (bi);
		\draw [blue, dashed] (a'i) -- node[edgelabel,  near start] {1} node[edgelabel,  near end] {2} (b'i);
		\draw [blue, dashed] (ck) -- node[edgelabel,  near start] {1} node[edgelabel,  near end] {$\infty$} (dk);
		\draw [blue, dashed] (c'k) -- node[edgelabel,  near start] {1} node[edgelabel,  near end] {2} (d'k);
		\draw [red, dotted] (ai) -- node[edgelabel,  near start] {4} node[edgelabel,  near end] {1} (b'i);
		\draw [red, dotted] (a'i) -- node[edgelabel,  near start] {2} node[edgelabel,  near end] {1} (bi);
		\draw [red, dotted] (ck) -- node[edgelabel,  near start] {2} node[edgelabel,  near end] {1} (d'k);
		\draw [red, dotted] (c'k) -- node[edgelabel,  near start] {2} node[edgelabel,  near end] {1} (dk);
		
		\draw [green] (ai) to[out=30,in=150] node[edgelabel, very near start] {2} node[edgelabel, very near end] {2} (dk);
		\draw [green!30!white] ($(dk) +(-0.6,0.6)$) to[out=-30,in=130] (dk);
		\draw [green!30!white] ($(dk) +(-0.9,0.3)$) to[out=-30,in=170] (dk);

		\end{tikzpicture}
\caption{Suppose $x$ occurs in clause $C_i$ (gadget on the left) and $\neg x$ occurs in clause $D_k$ (gadget on the right). The rank $\infty$ on the edge $(d_k,c_k)$ denotes that $c_k$ is $d_k$'s {\em last} choice neighbor and the rank $\infty-1$ on the edge $(d_k,u_k)$ denotes that $u_k$ is $d_k$'s {\em last but one} choice neighbor.} 
\label{newfig2:example}
\end{figure}
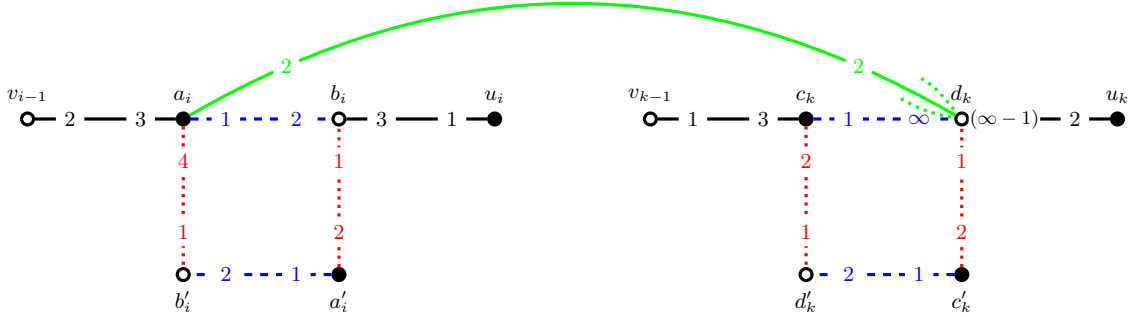

\begin{lemma}
\label{lem:unstable-edges}
		Let $S$ be any stable matching in $G,{\cal P}$. Then 
	\begin{enumerate}
            \item\label{enum:unstable-edges-1} $S$ leaves $s,t$ unmatched, $F\subseteq S$, and $S$ contains no consistency edge. \end{enumerate}
            Also, for every variable $x \in \{X_1,\ldots,X_{2n}\}$, we have:
            \begin{enumerate}
	\setcounter{enumi}{1}
    \item\label{enum:unstable-edges-2} From the gadget of $\neg x$, either (i)~$(c_k,d_k),(c'_k,d'_k) \in S$ or (ii)~$(c_k,d'_k),(c'_k,d_k) \in S$.  

--- If (i) happens, we say that the gadget is in \emph{true state}, else that it is in \emph{false state}. 
	\item\label{enum:unstable-edges-3} From a gadget of $x$ in, say, the $i$-th clause, either the pair of edges (i)~$(a_i,b'_i),(a'_i,b_i)$ or 
			(ii)~$(a_i,b_i),(a'_i,b'_i)$ is in $S$. 

--- If (i) happens, we say that the gadget is in \emph{true state}, else that it is in \emph{false state}. 
	\item\label{enum:unstable-edges-4} If the gadget of $\neg x$ is in true state then all the gadgets of $x$ are in false state. 
		\end{enumerate}
		
			\vspace{-.2cm}

\noindent Conversely, any matching $S$ of $G$ that satisfies 1-4 above is stable.
	\end{lemma}

	Lemma~\ref{lem:unstable-edges} implies that each stable matching can be mapped to a true/false assignment as follows: 
	\begin{itemize}
	    \item For each $x$, if the gadget of $\neg x$ is in true state, then set $x= \false$, else set $x = \true$. 
	\end{itemize}
	Conversely, for any $\true$/$\false$ assignment $A$ to the variables in $\phi$, we define an associated matching $M_A$ as follows. First, include all basic edges in $M_A$. 
	For any $x \in \{X_1,\ldots,X_{2n}\}$: 
	\begin{itemize}
	    \item  If $x = \true$ then set all gadgets corresponding to $x$ in true state, and the gadget corresponding to $\neg x$ in false state. Thus the {\em dotted red} edges (see Fig.~\ref{newfig2:example}) from all these gadgets get added to $M_A$.
	    \item Otherwise, set all gadgets corresponding to $x$ in false state and the gadget corresponding to $\neg x$ in true state. Thus the {\em dashed blue} edges (see Fig.~\ref{newfig2:example}) from all these gadgets get added to $M_A$.
	\end{itemize}
	
	Lemma~\ref{lem:unstable-edges} implies that $M_A$ is stable. The next fact concludes the reduction.
	\begin{lemma}
		\label{cl:stable-no-consistency} Let $S$ be a stable matching in $G,{\cal P}$. If there is an augmenting path $\rho$ in $G_{S}$, then $\rho$ goes from $s$ to $t$, does not use any consistency edge, and passes, in each clause, in exactly one gadget, which is in true state. In particular, the assignment corresponding to $S$ is feasible. Conversely, if $A$ is a feasible assignment, in each clause $i$ there is a gadget $Z_{i,j_i}$ in true state, and there exists an augmenting path in $G_{S}$ that passes through $Z_{i,j_i}$ for all $i$.
	\end{lemma}

	In the rest of the section, we give the missing details from the construction, and prove Lemma~\ref{lem:unstable-edges} and Lemma~\ref{cl:stable-no-consistency}.
	
	\subsection{The preference lists}\label{sec:preferences}
	
	The gadgets corresponding to an occurrence of $x$ in the $i$-th clause and the occurrence of $\neg x$ are given in 
Fig.~\ref{newfig2:example}. Let $D_k$ be the clause that contains the unique occurrence of $\neg x$ in the transformed formula $\phi$.

Note that the labels of vertices $a_i,b_i,a'_i,b'_i$ should depend on $x$; however, for the sake of readability, we omit the dependency on $x$, since it is always clear from the context. Similarly for the labels of vertices $c_k,d_k,c'_k,d'_k$.

We now describe the preference lists of the 4 vertices in the gadget of $x$ in the clause $C_i$. 

\begin{minipage}[c]{0.45\textwidth}
			
			\centering
			\begin{align*}
                          &  a_i : \ b_i \succ d_k \succ v_{i-1} \succ b'_i \qquad\qquad && a'_i : \ b'_i \succ b_i \\
                          &  b_i : \ a'_i \succ a_i \succ u_i  \qquad\qquad && b'_i : \ a_i \succ a'_i \\
			\end{align*}
\end{minipage}

We now describe the preference lists of the 4 vertices in the gadget of $\neg x$.

\begin{minipage}[c]{0.45\textwidth}
			
			\centering
			\begin{align*}
                          &  c_k : \ d_k \succ d'_k \succ v_{k-1}  \qquad\qquad && c'_k : \ d'_k \succ d_k \\
                          &  d_k : \ c'_k \succ a_i \succ a_j \succ \cdots \succ u_k \succ c_k  \qquad\qquad && d'_k : \ c_k \succ c'_k \\
			\end{align*}
\end{minipage}

Here $a_i,a_j,\ldots$ are the occurrences of the $a$-vertex in all the gadgets corresponding to literal $x$ in the formula $\phi$.
That is, the literal $x$ occurs in clauses $i,j,\ldots$ The order among the vertices $a_i,a_j,\ldots$ in $d_k$'s preference list does
not matter.

We now describe the preference lists of vertices $u_i$ and $v_i$. The preference list of $u_0$ is $v_0 \succ s$ and
for $1 \le i \le m+n$, the preference list of $u_i$ is as given on the left (these correspond to {\em positive} clauses $C_i$)
and for $m+n+1 \le i \le m+2n$, the preference list of $u_i$ is as given on the right (these correspond to {\em negative} clauses $D_i$):
\[ u_i: \ \ b_{i1} \succ b_{i2} \succ b_{i3} \succ \underline{v_i} \ \ \ \ \ \ \ \ \ \ \ \ \text{or}\ \ \ \ \ \ \ \ \ \ \ \ u_i: \ \ \underline{v_i} \succ d_{i1} \succ d_{i2},\]
where $b_{ij}$ (similarly, $d_{ij}$) is the $b$-vertex (resp., $d$-vertex) that appears in the gadget corresponding to the $j$-th literal in 
the $i$-th clause. If the $i$-th clause (a positive one) consists of only 2 literals then there is no $b_{i3}$ here. The vertex $v_i$ is underlined.

The preference list of $v_{m+2n}$ is $u_{m+2n} \succ t$ and 
for $1 \le i \le m+n$, the preference list of $v_{i-1}$ is as given on the left (these correspond to {\em positive} clauses $C_i$)
and for $m+n+1 \le i \le m+2n$, the preference list of $v_{i-1}$ is as given on the right (these correspond to {\em negative} clauses $D_i$):
 \[ v_{i-1}: \ \ \underline{u_{i-1}} \succ a_{i1} \succ a_{i2} \succ a_{i3} \ \ \ \ \ \ \ \ \ \ \ \ \text{or}\ \ \ \ \ \ \ \ \ \ \ \ v_{i-1}: \ \ c_{i1} \succ c_{i2}  \succ \underline{u_{i-1}},\]
where $a_{ij}$ (similarly, $c_{ij}$) is the $a$-vertex (resp., $c$-vertex) that appears in the gadget corresponding to the $j$-th literal 
in the $i$-th clause. If the $i$-th clause (a positive one) consists of only 2 literals then there is no $a_{i3}$ here. The vertex $u_{i-1}$
is underlined.
	
	\subsection{Proof of Lemma~\ref{lem:unstable-edges}}
	
	We first derive some useful properties. Let $S$ be the matching given by the union of the basic set and the pair of dashed blue edges from the gadget of every literal. One easily checks that $S$ is stable. Thus $s$ and $t$ are unstable vertices, i.e., they remain unmatched in any stable matching. Suppose an edge $e$ is labeled $(-,-)$ with respect to some stable matching. It is an easy fact~\cite{GI89} that $e$ is an {\em unstable edge}, i.e., no stable matching contains $e$.
The following lemma is based on this fact.

\begin{lemma}
  \label{lem:unstable-edges-aux}
  Let $S$ be any stable matching in $G,{\cal P}$. Then (i)~$S$ does not contain any consistency edge and (ii) $S$ is a superset of the basic set.
\end{lemma}
\begin{proof}
  Consider the matching $N$ given by the union of the basic set, the pair of dashed blue edges from the gadget of every positive literal, and the pair of dotted red edges from the gadget of every negative literal. It is easy to see that $N$ is a stable matching. Every consistency edge is labeled $(-,-)$ with respect to $N$. Thus every consistency edge is unstable, proving (i).
  
  Observe that every edge $(v_{i-1},a_i)$ is labeled $(-,-)$ with respect to $N$. Similarly, every edge $(d_k,u_k)$ is also labeled $(-,-)$ with respect to $N$. Thus no stable
  matching can contain these edges. Recall that all the $u$-vertices and $v$-vertices are stable, so this immediately implies that $v_{i-1}$ is matched to $u_{i-1}$ for all $i \le m+n$ and $u_k$ is matched $v_k$ for all $k \ge m+n+1$. Also recall that the 4 vertices in each literal gadget are stable.
  This implies $u_{m+n}$ and $v_{m+n}$ are also matched to each other in every stable matching, proving (ii). \qed
\end{proof}

We can now conclude the proof of Lemma~\ref{lem:unstable-edges}. Let $S$ be any stable matching in $G$. Statement~\ref{enum:unstable-edges-1} follows from Lemma~\ref{lem:unstable-edges-aux}. Statements~\ref{enum:unstable-edges-2} and~\ref{enum:unstable-edges-3} are shown by statement~\ref{enum:unstable-edges-1} and stability. As for statement~\ref{enum:unstable-edges-4}, if
the dashed blue pair of edges from $\neg x$'s gadget is included in $S$ then the stability of $S$ implies that $S$ has to contain the dashed blue pair of edges from every gadget of $x$---otherwise some consistency edge would be a {\em blocking edge} to $S$. Finally, it is easy to see that any matching that satisfies statements \ref{enum:unstable-edges-1}-\ref{enum:unstable-edges-4} is stable.

	\subsection{Proof of Lemma~\ref{cl:stable-no-consistency}}
	
	We first show the second part of Lemma~\ref{cl:stable-no-consistency}. Let $S$ be a stable matching of $G$. If there is a gadget in true state in each clause, then an $M$-augmenting path starting at $s$ and ending at $t$ in $G_M$ is easily constructed as follows. First, take edge $(s,u_0)$, then all edges $(u_i,v_i)$ and: 
	\begin{itemize}
	\item for each clause $C_i$ with gadget $Z$ set to true, edges $(v_{i-1},a_i), (a_i,b_i'), (b_i',a_i'), (a_i',b_i), (b_i,u_i)$; 
	\item for each clause $D_k$ with gadget $Z$ set to true, edges $(v_{k-1},c_k), (c_k,d_k), (d_k,u_k)$. 
	\end{itemize}
	Last, take edge $(v_{m+2n},t)$. In order to prove the other direction, we start with some auxiliary facts.

\begin{new-claim}
\label{claim0}
All popular matchings in $G$ have the same size and match the same set of vertices. In particular, they leave unmatched only $s$ and $t$.
\end{new-claim}
\begin{proof}
  Consider the matching $M$ given by the union of the basic set with the pair of dashed blue edges from the gadget of every literal.
Note that $M$ is a stable matching in $G$.   We claim that $M$ is also a dominant matching in $G$. Since a dominant (resp. stable) matching is a popular matching of maximum (resp. minimum) size, the first claim follows. The second is then implied by Lemma~\ref{lem:all-same}. 
  Consider the edge $(v_0,a_1)$: this is labeled $(-,-)$ with respect to $M$ since both $v_0$ and $a_1$ prefer
  their partners in $M$ to each other. Thus there is no $s$-$t$ path in the graph $G_{M}$. As $s$ and $t$ are the only unmatched vertices, 
  there is no $M$-augmenting path in $G_{M}$. So $M$ is a dominant matching in $G$ by Theorem~\ref{thm:dominant}. \qed
\end{proof}

\begin{new-claim}
\label{cl:inproof-stable-no-consistency} Let $S$ be a stable matching in $G$ and $\rho$ an augmenting path in $G_{S}$. Then $\rho$ goes from $s$ to $t$ and does not use any consistency edge.
\end{new-claim}
\begin{proof}
By Claim~\ref{claim0}, stable matchings in $G$ leave only 2 vertices unmatched: these are $s$ and $t$. So $\rho$ must go from $s$ to~$t$. Recall that consistency edges cannot be included in $S$, see Lemma~\ref{lem:unstable-edges}. By parity reasons, traversing $\rho$ from $s$ to $t$, a consistency edge can only occur in $\rho$ if it leads $\rho$ back to an earlier clause. That is, a consistency edge $(a_i,d_k)$ has to be traversed in $\rho$ in the direction $d_k \rightarrow a_i$. Let $e = (d_k,a_i)$ be the first consistency edge traversed in $\rho$.
Observe that $\rho$ cannot reach $t$, because nodes $v_{i-1}$ and $u_i$ must have already been traversed by $\rho$. Hence, consistency edges cannot appear in~$\rho$. \qed
\end{proof}

We can now conclude the proof of Lemma~\ref{cl:stable-no-consistency}.
Assume that there is an augmenting path $\rho$ in $G_S$, starting at $s$ and terminating at $t$. Because of Lemma~\ref{cl:stable-no-consistency}, $\rho$ goes from $s$ to $t$, traverses all clauses, and for each clause $i$, there is a path in $G_S$ between $v_{i-1}$ and $u_i$. The literal whose gadget provides this $v_{i-1} \rightarrow u_i$ path is set to $\true$ and thus the $i$-th clause is satisfied. As this holds for each $i$, 
this means that $\phi$ has a satisfying assignment. 

\subsection{Max-size popular matchings}

A non-dominant popular matching trivially exists if the size of a stable matching differs from the size of a dominant matching in an instance. Our next result is tailored for such instances. We will now show that it is NP-hard to decide if $G\new{, \mathcal{P}}$ admits a {\em max-size} popular matching that is not dominant. 

\begin{pr}
    \inp A bipartite graph $G = (A \cup B,E)$ with strict preference lists.\\
	\ques If there is a non-dominant matching among max-size popular matchings in~$G$.
\end{pr}

Given a 3SAT formula $\phi$, we will transform it as described earlier and build the graph $G_{\phi}$. Recall that all popular matchings in $G_{\phi}$ have the same size.
We proved in Theorem~\ref{thm:non-dominant} that it is NP-hard to decide if $G_{\phi}$ admits a popular matching that is not dominant. Consider the graph $H = G_{\phi} \cup \rho$ where $\rho$ is the path $\langle p_0,q_0,p_1,q_1\rangle$ with $p_1$ and $q_0$ being each other's top choices. There are no edges between a vertex in $\rho$ and a vertex in $G_{\phi}$. A max-size popular matching in $H$ consists of the pair of edges $(p_0,q_0), (p_1,q_1)$, and a popular matching in $G_{\phi}$. 

Hence a max-size popular matching in $H$ that is {\em not} dominant consists of $(p_0,q_0), (p_1,q_1)$, and a popular matching in $G_{\phi}$ that is non-dominant. Thus the desired result immediately follows.

\begin{theorem}
  \label{thm:max-size}
Given $G = (A \cup B,E)$ \new{ and $\mathcal{P}$} where in \new{$\mathcal{P}$}, every vertex has a strict ranking over its neighbors, it is NP-complete to decide if $G$ admits a max-size popular matching that is not dominant.
\end{theorem}

\section{Hardness of finding a stable matching that is dominant and consequences for popular matchings}\label{sec:new}
Given $G = (A \cup B, E)$ \new{and a set of strictly ordered lists $\mathcal{P}$}, we first consider here the problem of deciding if $G$ admits a matching that is {\em both} stable and dominant.

\begin{pr}\label{pr:stable-cap-dominant}
    \inp A bipartite graph $G = (A \cup B,E)$ with strict preference lists.\\
	\ques If there is a stable matching in $G$ that is also dominant.
\end{pr}

In instances where all popular matchings have the same size, such a matching $M$ is desirable as there are no blocking edges with respect to $M$; moreover, $M$ defeats any larger matching in a head-to-head election. 

In Section~\ref{sec:stable-dominant}, we show that the above problem is NP-complete.
We will then use this hardness to show the NP-completeness of deciding if $G$ admits a min-size popular matching that is unstable and also to give a short proof of NP-hardness of the popular roommates problem. 

\subsection{Finding a matching that is both stable and dominant}
\label{sec:stable-dominant}

We know from Theorem~\ref{thm:dominant} that Problem~\ref{pr:stable-cap-dominant} is equivalent to the problem of deciding if there exists a stable matching $M$ such that $M$ has no augmenting path in $G_M$. We will show a simple reduction from 3SAT to this stable matching problem. Note that the graph $G$ here (and hence the reduction) is different from the instance given in Section~\ref{sec:stable}. Given a 3SAT formula $\phi$, we will transform $\phi$ as done in Section~\ref{sec:stable} so that there is exactly one occurrence of $\neg x$ in $\phi$ for every variable $x$. 

Corresponding to the above formula $\phi$, we will construct an instance $G\new{, \mathcal{P}}$ whose high-level picture is shown in Fig.~\ref{newfig3:example}. 
There are two ``unwanted'' vertices $s, t$ along with $u^{\ell}_j,v^{\ell}_j$ for every clause $\ell$ in $\phi$, for $0 \le j \le i$, where $i$ is the number of literals in clause $\ell$, along with one gadget for each literal in every clause. 

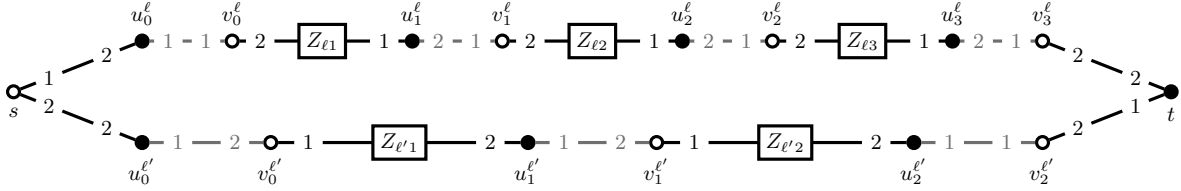
\begin{figure}[ht]
	\centering
		\begin{tikzpicture}[scale=0.9, transform shape, very thick]
		\pgfmathsetmacro{\b}{1.9}
		\pgfmathsetmacro{\d}{1.5}
		
		\node[vvertex, label=below:$s$] (s) at (1*\b, \d*0.5) {};
		\node[uvertex, label=below:$u_0^{\ell'}$] (u0) at (2*\b, 0) {};
		\node[vvertex, label=below:$v_0^{\ell'}$] (v0) at (3*\b, 0) {};
		\node[uvertex, label=below:$u_1^{\ell'}$] (u1) at (5*\b, 0) {};
		\node[vvertex, label=below:$v_1^{\ell'}$] (v1) at (6*\b, 0) {};
		\node[uvertex, label=below:$u_2^{\ell'}$] (u2) at (8*\b, 0) {};
		\node[vvertex, label=below:$v_2^{\ell'}$] (v2) at (9*\b, 0) {};
		
		\node[uvertex, label=above:$u_0^{\ell}$] (u0') at (2*\b, \d) {};
		\node[vvertex, label=above:$v_0^{\ell}$] (v0') at (2.7*\b, \d) {};
		\node[uvertex, label=above:$u_1^{\ell}$] (u1') at (4.1*\b, \d) {};
		\node[vvertex, label=above:$v_1^{\ell}$] (v1') at (4.8*\b, \d) {};
		\node[uvertex, label=above:$u_2^{\ell}$] (u2') at (6.2*\b, \d) {};
		\node[vvertex, label=above:$v_2^{\ell}$] (v2') at (6.9*\b, \d) {};
		\node[uvertex, label=above:$u_3^{\ell}$] (u3') at (8.3*\b, \d) {};
		\node[vvertex, label=above:$v_3^{\ell}$] (v3') at (9*\b, \d) {};
		
		\node[uvertex, label=below:$t$] (t) at (10*\b, \d*0.5) {};
				
		\draw (s) -- node[edgelabel, near start] {2} node[edgelabel, near end] {2} (u0);
		\draw [MyPurple] (u0) -- node[edgelabel, near start] {1} node[edgelabel, near end] {2} (v0);
		\draw (v0) -- node[edgelabel, very near start] {1} node[edgebox] {$Z_{\ell' 1}$} node[edgelabel, very near end] {2} (u1);
		\draw [MyPurple] (u1) -- node[edgelabel, near start] {1} node[edgelabel, near end] {2} (v1);
		\draw (v1) -- node[edgelabel, very near start] {1} node[edgebox] {$Z_{\ell' 2}$} node[edgelabel, very near end] {2} (u2);
		\draw [MyPurple] (u2) -- node[edgelabel, near start] {1} node[edgelabel, near end] {1} (v2);
		\draw (v2) -- node[edgelabel, near start] {2} node[edgelabel, near end] {1} (t);
		
		\draw (s) -- node[edgelabel, near start] {1} node[edgelabel, near end] {2} (u0');
		\draw [MyPurple] (u0') -- node[edgelabel, near start] {1} node[edgelabel, near end] {1} (v0');
		\draw (v0') -- node[edgelabel, very near start] {2} node[edgebox] {$Z_{\ell 1}$} node[edgelabel, very near end] {1} (u1');
		\draw [MyPurple] (u1') -- node[edgelabel, near start] {2} node[edgelabel, near end] {1} (v1');
		\draw (v1') -- node[edgelabel, very near start] {2} node[edgebox] {$Z_{\ell 2}$} node[edgelabel, very near end] {1} (u2');
		\draw [MyPurple] (u2') -- node[edgelabel, near start] {2} node[edgelabel, near end] {1} (v2');
		\draw (v2') -- node[edgelabel, very near start] {2} node[edgebox] {$Z_{\ell 3}$} node[edgelabel, very near end] {1} (u3');
		\draw [MyPurple] (u3') -- node[edgelabel, near start] {2} node[edgelabel, near end] {1} (v3');
		\draw (v3') -- node[edgelabel, near start] {2} node[edgelabel, near end] {2} (t);
		
		\end{tikzpicture}
\caption{The high-level picture of the gadgets corresponding to clauses $\ell$ with three literals, and $\ell'$ with two negated literals in the instance $G\new{, \mathcal{P}}$. The vertices $s$ and $t$ are common to all clauses.}
\label{newfig3:example}
\end{figure}

In Fig.~\ref{newfig3:example}, $Z_{{\ell}j}$ is the gadget corresponding to the $j$-th literal of clause $\ell$, where each clause has 2 or 3 literals. As before,  we have a separate gadget for every occurrence of each literal. That is, there is a separate gadget for each occurrence of $x \in \{X_1,\ldots,X_{2n}\}$ in $\phi$ and another gadget for the unique occurrence of $\neg x$ in $\phi$. 

As done in Section~\ref{sec:stable}, vertices from a gadget corresponding to $x$ are denoted by $a_i,a_i',b_i,b_i'$, while vertices from a gadget corresponding to $\neg x$ are denoted by $c_k,c_k',d_k,d'_k$; however adjacency lists and preferences are different here, see again Fig.~\ref{newfig3:example}. As before, we have a {\em consistency edge} between $x$'s gadget and $\neg x$'s gadget: this is now between the $b$-vertex of $x$'s gadget and $c$-vertex of $\neg x$'s gadget (see Fig.~\ref{newfig4:example}). Again, we postpone the definition of preference lists, and instead state the following lemma, analogous to Lemma~\ref{lem:unstable-edges}. The proof of Lemma~\ref{lem:unstable-edges-bis} is analogous to the proof of Lemma~\ref{lem:unstable-edges} and hence omitted.

\begin{figure}[ht]
	\centering
		\begin{tikzpicture}[scale=0.9, transform shape, very thick]
		\pgfmathsetmacro{\b}{2.3}
		\node[vvertex, label=above:$v^i_0$] (vi-1) at (0, 0) {};
		\node[uvertex, label=above:$a_1$] (ai) at (1*\b, 0) {};
		\node[vvertex, label=above:$b_1$] (bi) at (2*\b, 0) {};
		\node[uvertex, label=above:$u^i_1$] (ui) at (3*\b, 0) {};
		\node[vvertex, label=above:$v^k_1$] (vk-1) at (4*\b, 0) {};
		\node[uvertex, label=above:$c_2$] (ck) at (5*\b, 0) {};
		\node[vvertex, label=above:$d_2$] (dk) at (6*\b, 0) {};
		\node[uvertex, label=above:$u^k_2$] (uk) at (7*\b, 0) {};
		\node[vvertex, label=below:$b'_1$] (b'i) at (1*\b, -\b) {};
		\node[uvertex, label=below:$a'_1$] (a'i) at (2*\b, -\b) {};
		\node[vvertex, label=below:$d'_2$] (d'k) at (5*\b, -\b) {};
		\node[uvertex, label=below:$c'_2$] (c'k) at (6*\b, -\b) {};
		
		\draw (vi-1) -- node[edgelabel,  near start] {2} node[edgelabel,  near end] {2} (ai);
		\draw (bi) -- node[edgelabel,  near start] {4} node[edgelabel,  near end] {1} (ui);
		\draw (vk-1) -- node[edgelabel,  near start] {1} node[edgelabel,  near end] {$\infty$} (ck);
		\draw (dk) -- node[edgelabel,  near start] {2} node[edgelabel,  near end] {2} (uk);		
		
		\draw [blue, dashed] (ai) -- node[edgelabel,  near start] {1} node[edgelabel,  near end] {3} (bi);
		\draw [blue, dashed] (a'i) -- node[edgelabel,  near start] {1} node[edgelabel,  near end] {2} (b'i);
		\draw [blue, dashed] (ck) -- node[edgelabel,  near start] {1} node[edgelabel,  near end] {3} (dk);
		\draw [blue, dashed] (c'k) -- node[edgelabel,  near start] {1} node[edgelabel,  near end] {2} (d'k);
		\draw [red, dotted] (ai) -- node[edgelabel,  near start] {3} node[edgelabel,  near end] {1} (b'i);
		\draw [red, dotted] (a'i) -- node[edgelabel,  near start] {2} node[edgelabel,  near end] {1} (bi);
		\draw [red, dotted] (ck) -- node[edgelabel,  near start] {$\infty-1$} node[edgelabel,  near end] {1} (d'k);
		\draw [red, dotted] (c'k) -- node[edgelabel,  near start] {2} node[edgelabel,  near end] {1} (dk);
		
		\draw [green] (bi) to[out=30,in=150] node[edgelabel, very near start] {2} node[edgelabel, very near end] {2} (ck);
		\draw [green!30!white] ($(ck) +(-0.6,0.6)$) to[out=-30,in=130] (ck);
		\draw [green!30!white] ($(ck) +(-0.9,0.3)$) to[out=-30,in=170] (ck);

		\end{tikzpicture}
\caption{Suppose $x$ occurs as the first literal in the $i$-th clause and $\neg x$ occurs as the second literal in the $k$-th clause. For convenience, we use $a_1,b_1,a'_1,b'_1$ to denote the 4 vertices in this gadget of $x$ and $c_2,d_2,c'_2,d'_2$ to denote the 4 vertices in the gadget of $\neg x$. As before, $\infty$ and $\infty-1$ denote last choice neighbor and last but one choice neighbor, respectively.} 
\label{newfig4:example}
\end{figure}
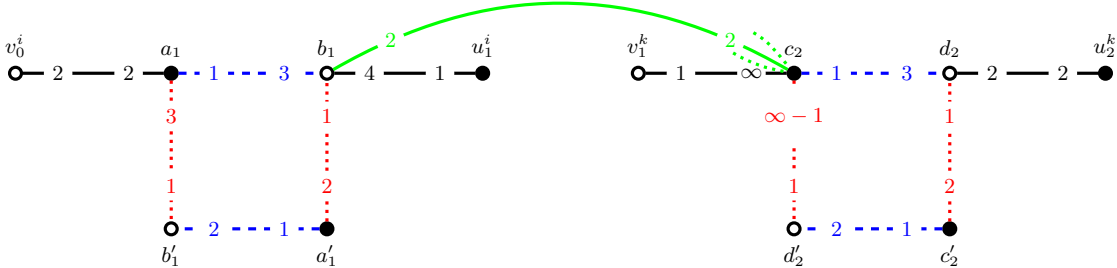

\begin{lemma}
\label{lem:unstable-edges-bis}
		Let $S$ be any stable matching in $G,{\cal P}$. Then 
\begin{enumerate}

\item\label{enumbis:unstable-edges-1} $S$ leaves $s$ and $t$ unmatched and it does not contain any consistency edge, 
while the edges $(u^{\ell}_i,v^{\ell}_i)$ for all clauses $\ell$ and indices $i$ are in $S$. \end{enumerate}
Also, for every variable $x \in \{X_1,\dots, X_{2n}\}$, we have:
\begin{enumerate}
\setcounter{enumi}{1}
\item\label{enumbis:unstable-edges-2} From the gadget of $\neg x$, either the pair (i)~$(c_k,d_k),(c'_k,d'_k)$ or (ii)~$(c_k,d'_k),(c'_k,d_k)$ is in $S$. 

--- If (i) happens, we say that the gadget is in \emph{false state}, else that it is in \emph{true state}. 
\item\label{enumbis:unstable-edges-3} From a gadget of $x$, say, its gadget in the $i$-th clause, either the pair of edges (i)~$(a_i,b_i),(a'_i,b'_i)$ or 
(ii)~$(a_i,b'_i),(a'_i,b_i)$ is in $S$.

--- If (i) happens, we say that the gadget is in \emph{true state}, else that it is in \emph{false state}. 
\item\label{enumbis:unstable-edges-4} If the gadget of $\neg x$ is in true state then all the gadgets of $x$ are in false state. \end{enumerate}
		

\noindent Conversely, any matching $S$ of $G$ that satisfies 1-4 above is stable.
	\end{lemma}

\noindent Corresponding to any stable matching $M$, define a $\true$/$\false$ assignment $A_M$ as follows:
\begin{itemize}
\item If the dashed blue pair $(c_k,d_k)$ and $(c'_k,d'_k)$ from $\neg x$'s gadget belongs to $M$ then 
$A_M$ sets $x$ to $\true$; else $A_M$ sets $x$ to $\false$.
\end{itemize}

Conversely, for any $\true$/$\false$ assignment $A$ to the variables in $\phi$, we define an associated matching $M_A$ as follows. For any $x \in \{X_1,\ldots,X_{2n}\}$:
\begin{itemize}
    \item If $x = \true$, then set every gadget corresponding to $x$ in true state, and the gadget corresponding to $\neg x$ in false state.
    Thus the {\em dashed blue} edges (see Fig.~\ref{newfig4:example}) from all these gadgets get added to $M_A$. 
    \item Otherwise, set every gadget corresponding to $x$ in false state, and the gadget corresponding to $\neg x$ in true state. 
    Thus the {\em dotted red} edges (see Fig.~\ref{newfig4:example}) from all these gadgets get added to $M_A$.
\end{itemize}
Finally, include the edges $(u^{\ell}_i,v^{\ell}_i)$ for all $\ell$ and $i$.
Using Lemma~\ref{lem:unstable-edges-bis}, it is easy to see that $M_A$ is a stable matching.

Our goal is to show that $\phi$ is satisfiable if and only if $G$ has a stable matching $S$ without an $S$-augmenting path in $G_S$. This is achieved by the following lemma, whose proof is given at the end of this section.

	\begin{lemma}
		\label{cl:stable-no-consistency-bis} If there is a stable matching $M$ such that $G_M$ has no augmenting path, then $A_M$ satisfies $\phi$. Conversely, if there exists a satisfiable assignment $A$, then there is no augmenting path with respect to $M_A$ in $G_{M_A}$.
	\end{lemma}

We can therefore conclude the following.

\begin{theorem}
  \label{thm:stable-and-dominant}
Given $G = (A \cup B,E),{\cal P}$, it is NP-complete to decide if $G$ admits a matching that is both stable and dominant.
\end{theorem}

\subsubsection{Preference lists}\label{sec:hardn2-lists}

We now describe the preference lists of the 4 vertices $a_1,b_1,a'_1,b'_1$ in the gadget of $x$ in the $i$-th clause. 
We assume $x$ to be the first literal in this clause.
Note that the consistency edge now exists between the $b$-vertex in $x$'s gadget and the $c$-vertex in $\neg x$'s gadget 
(see Fig.~\ref{newfig4:example}).

\begin{minipage}[c]{0.45\textwidth}
			
			\centering
			\begin{align*}
                          &  a_1 : \ b_1 \succ v^i_0 \succ b'_1 \qquad\qquad && a'_1 : \ b'_1 \succ b_1 \\
                          &  b_1 : \ a'_1 \succ c_k \succ a_1 \succ u^i_1  \qquad\qquad && b'_1 : \ a_1 \succ a'_1 \\
			\end{align*}
\end{minipage}

We now describe the preference lists of the 4 vertices $c_2,d_2,c'_2,d'_2$ in the gadget of $\neg x$ (see Fig.~\ref{newfig4:example}).
We assume $\neg x$ to be the second literal in this clause (let this be the $k$-th clause).

\begin{minipage}[c]{0.45\textwidth}
			
			\centering
			\begin{align*}
                          &  c_2 : \ d_2 \succ b_i \succ b_j \succ \cdots \succ d'_2 \succ v^k_1  \qquad\qquad && c'_2 : \ d'_2 \succ d_2\\
                          &  d_2 : \ c'_2  \succ u^k_2 \succ c_2  \qquad\qquad && d'_2 : \ c_2 \succ c'_2 \\
			\end{align*}
\end{minipage}

Here $b_i,b_j,\ldots$ are the occurrences of the $b$-vertex in all the gadgets corresponding to literal $x$ in the formula $\phi$.
The order among these vertices in $c_2$'s preference list is not important.

We now describe the preference lists of vertices $u^{\ell}_i$ and $v^{\ell}_i$ for $0 \le i \le r$ in clause~$\ell$,
where $r \in \{2,3\}$ is the number of literals in this clause. 
The preference list of $u^{\ell}_0$ is $v^{\ell}_0 \succ s$
and the preference list of $v^{\ell}_r$ is $u^{\ell}_r \succ t$.

Let $1 \le j \le r$ and let $a_j,b_j,a'_j,b'_j$ be the 4 vertices in the gadget of the $j$-th literal in this clause if it is a positive clause and let $c_j,d_j,c'_j,d'_j$ be the 4 vertices in the gadget of the $j$-th literal in this clause if it is a negative clause. The preference list of $u^{\ell}_j$ is as given on the left (resp. right) for positive (resp. negative) clauses.
\[ u^{\ell}_j: \ \ b_j \succ v^{\ell}_j \ \ \ \ \ \ \ \ \ \ \ \ \text{or}\ \ \ \ \ \ \ \ \ \ \ \ u^{\ell}_j: \ \ v^{\ell}_j \succ d_j.\]
The preference list of $v^{\ell}_{j-1}$ is as given on the left (resp. right) for positive (resp. negative) clauses.
\[ v^{\ell}_{j-1}: \ \ u^{\ell}_{j-1} \succ a_j \ \ \ \ \ \ \ \ \ \ \ \ \text{or}\ \ \ \ \ \ \ \ \ \ \ \ v^{\ell}_{j-1}: \ \ c_j \succ u^{\ell}_{j-1}.\]

Observe that $s$ and $t$ are the last choices for each of their neighbors. The preferences of $s$ and $t$ are not relevant and it is easy to see the following.

\begin{new-claim}\label{cl:hardn2} Both $s$ and $t$ are unstable vertices. 
\end{new-claim}

\subsubsection{Proof of Lemma~\ref{cl:stable-no-consistency-bis}}\label{sec:hardn2-secondlemma}

Take any stable matching $M$. Suppose $A_M$ does not satisfy, say, the $r$-th clause. We show that $G_M$ has an augmenting path with respect to $M$, concluding the proof of one direction. Note that all gadgets in the $r$-th clause are in false state in $M$. Construct the augmenting path in $G_M$ follows:
\begin{itemize}
\item Let the $r$-th clause be a positive clause $x \vee y \vee z$. Let $a_1,b_1,a'_1,b'_1$ be the 4 vertices in $x$'s gadget,
  $a_2,b_2,a'_2,b'_2$ be the 4 vertices in $y$'s gadget, and  $a_3,b_3,a'_3,b'_3$ be the 4 vertices in $z$'s gadget. Thus we have the
  following augmenting path with respect to $M$:
  \[s - (u^r_0,v^r_0) - {\color{red}(a_1,b'_1)} - {\color{red}(a'_1,b_1)} - (u^r_1,v^r_1) - {\color{red}(a_2,b'_2)} - {\color{red}(a'_2,b_2)} - (u^r_2,v^r_2) - {\color{red}(a_3,b'_3)} - {\color{red}(a'_3,b_3)} - (u^r_3,v^r_3) - t.\]

\item Let the $r$-th clause be a negative clause $\neg x \vee \neg y$. So both $x$ and $y$ are set to $\true$ and
  $M$ contains the dashed blue pair of edges from $\neg x$'s gadget and also from $\neg y$'s gadget.

  Let $c_1,d_1,c'_1,d'_1$ be the 4 vertices in $\neg x$'s gadget and let $c_2,d_2,c'_2,d'_2$ be the 4 vertices in $\neg y$'s gadget.
  Thus we have the following augmenting path with respect to $M$:
  \[s - (u^r_0,v^r_0) - {\color{blue}(c_1,d_1)} - (u^r_1,v^r_1) - {\color{blue}(c_2,d_2)} - (u^r_2,v^r_2) - t.\] 
\end{itemize}

Conversely, assume that $\phi$ has a satisfiable assignment $A$. 
We show that there is no $M_A$-augmenting path in $G_{M_A}$.
Let $\ell$ be any positive clause. Suppose $A$ sets the $j$-th literal in this clause to $\true$. Let $a_j,b_j,a'_j,b'_j$ be the
4 vertices in the gadget corresponding to the $j$-th literal. So $(a_j,b_j) \in M_A$ and hence there is no edge $(v^{\ell}_{j-1},a_j)$
in $G_{M_A}$. Thus there is no alternating path between $s$ and $t$ such that all the intermediate vertices on this path correspond to the $\ell$-th clause.

However there are consistency edges jumping across clauses---so there may be an $M_A$-augmenting path that begins with vertices
corresponding to the $\ell$-th clause and then uses a consistency edge. Let $\rho$ be an $M_A$-alternating path in $G_{M_A}$ with $s$ as an
endpoint and let $(s,u^{\ell}_0)$ be the first edge in~$\rho$. So the prefix of $\rho$ consists of vertices that belong to the $\ell$-th clause.

Let $b_i$ be the last vertex of the $\ell$-th clause in this prefix of $\rho$. Let $a_i,b_i,a'_i,b'_i$ be the 4 vertices in the gadget of $b_i$
(let this correspond to variable $x$) in the $\ell$-th clause. Observe that $(a_i,b'_i)$ and $(a'_i,b_i)$ are in $M_A$---otherwise $b_i$ is
not reachable from $u^{\ell}_0$ in $G_{M_A}$ via a path of vertices in the $\ell$-th clause.

Assume the consistency edge $(b_i,c_k)$ belongs to $\rho$, where $c_k$ is a vertex in the gadget of $\neg x$.
Let $c_k,d_k,c'_k,d'_k$ be the vertices in the gadget of $\neg x$ and suppose this literal occurs in the $r$-th clause.
Since the dotted red pair in $x$'s gadget in the $\ell$-th clause is in $M_A$, the dotted red pair $(c_k,d'_k)$ and $(c'_k,d_k)$ has to be in $M_A$.
Since $c'_k$ and $d'_k$ are degree~2 vertices, $\rho$ has to contain the subpath $c_k - d'_k - c'_k - d_k$. However $\rho$ is now stuck at $d_k$ in the
graph $G_{M_A}$. The path cannot go back to $c_k$. There is no edge between $d_k$ and $u^r_k$ in $G_{M_A}$ since this is a $(-,-)$ edge as
both  $d_k$ and $u^r_k$ prefer their
respective partners in $M_A$ to each other. Thus the alternating path $\rho$ has to terminate at $d_k$.

The case when $\ell$ is a negative clause is even simpler since in this case $\rho$ cannot leave the vertices of the $\ell$-th clause using a consistency
edge. We know that the assignment $A$ sets some literal in the $\ell$-th clause to $\true$: let this be the $k$-th literal,
so $(c'_k,d_k) \in M_A$ and hence there is no edge $(d_k,u^{\ell}_k)$ in $G_{M_A}$. Thus there is no $M_A$-augmenting path in $G_{M_A}$ in this case also, concluding the proof. \qed

\subsection{Min-size popular matchings}
\label{sec:min-size}
In this section we investigate the counterpart of Theorem~\ref{thm:max-size}, i.e., the complexity of determining if $G = (A\cup B,E)$ admits a
{\em min-size} popular matching that is not stable. 

\begin{pr}
    \inp A bipartite graph $G = (A \cup B,E)$ \new{with strict preference lists}.\\
	\ques If there is an unstable matching among min-size popular matchings in~$G$.
\end{pr}

Given a 3SAT formula $\phi$, we transform it as described in Section~\ref{sec:stable} and build the graph $G$ as described in
Section~\ref{sec:stable-dominant}.  We now augment the bipartite graph $G$ into bipartite graph $H$ as follows:
\begin{itemize}
\item Add a new vertex $w$, which is adjacent to each $d'$-vertex in $\neg x$'s gadget for every variable $x \in\{X_1,\ldots,X_{2n}\}$.

\item Add a square $\langle t,t',r',r\rangle$ at the $t$-end of the graph (see Fig.~\ref{newfig5:example}), where $t',r',r$ are new vertices.
\end{itemize}  

\begin{figure}[ht]
	\centering
		\begin{tikzpicture}[scale=0.9, transform shape, very thick]
		\pgfmathsetmacro{\b}{1.6}
		\pgfmathsetmacro{\d}{1.5}
		
		\node[vvertex, label=below:$s$] (s) at (1*\b, \d*0.5) {};
		\node[uvertex] (u0) at (2*\b, 0) {};
		\node[vvertex] (v0) at (3*\b, 0) {};
		\node[uvertex] (u1) at (5*\b, 0) {};
		\node[vvertex] (v1) at (6*\b, 0) {};
		\node[uvertex] (u2) at (8*\b, 0) {};
		\node[vvertex] (v2) at (9*\b, 0) {};

		\node[uvertex] (u0') at (2*\b, \d) {};
		\node[vvertex] (v0') at (2.7*\b, \d) {};
		\node[uvertex] (u1') at (4.1*\b, \d) {};
		\node[vvertex] (v1') at (4.8*\b, \d) {};
		\node[uvertex] (u2') at (6.2*\b, \d) {};
		\node[vvertex] (v2') at (6.9*\b, \d) {};
		\node[uvertex] (u3') at (8.3*\b, \d) {};
		\node[vvertex] (v3') at (9*\b, \d) {};
		
		\node[uvertex, label={below, xshift=-1mm}:$t$] (t) at (10*\b, \d*0.5) {};
				
		\draw (s) -- (u0);
		\draw [MyPurple] (u0) -- (v0);
		\draw (v0) -- node[edgebox] {} (u1);
		\draw [MyPurple] (u1) -- (v1);
		\draw (v1) -- node[edgebox] {} (u2);
		\draw [MyPurple] (u2) -- (v2);
		\draw (v2) --  (t);
		
		
		\draw (s) -- (u0');
		\draw [MyPurple] (u0') -- (v0');
		\draw (v0') -- node[edgebox] {} (u1');
		\draw [MyPurple] (u1') -- (v1');
		\draw (v1') -- node[edgebox] {} (u2');
		\draw [MyPurple] (u2') -- (v2');
		\draw (v2') -- node[edgebox] {} (u3');
		\draw [MyPurple] (u3') -- (v3');
		\draw (v3') -- (t);
		
		\node[vvertex, label=below:$t'$] (t') at (11*\b, 0) {};
		\node[vvertex, label=above:$r$] (r) at (11*\b, \d) {};
		\node[uvertex, label=below:$r'$] (r') at (12*\b, \d*0.5) {};

		\draw (t) -- node[edgelabelr,  near start] {$\infty-1$} node[edgelabel,  near end] {2} (r);
		\draw (r) -- node[edgelabel,  near start] {1} node[edgelabel,  near end] {1} (r');
		\draw (r') -- node[edgelabel,  near start] {2} node[edgelabel,  near end] {1} (t');
		\draw (t') -- node[edgelabel,  near start] {2} node[edgelabelr,  near end] {$\infty$} (t);

		\node[uvertex, label=below:$w$] (w) at (10*\b, 0) {};
		
		\draw (w) to[out=-150,in=-30] node[edgelabel, very near start] {2} node[edgelabel, very near end] {$\infty$} ($(v0)!0.5!(u1) +(-0.25,-0.24)$);
		\draw (w) to[out=-165,in=-15] node[edgelabel, very near start] {1} node[edgelabel, very near end] {$\infty$} ($(v1)!0.5!(u2) +(-0.26,-0.24)$);
		
		\end{tikzpicture}
\caption{We add a square $\langle t,t',r',r\rangle$ at the end of the graph $G$, and a vertex $w$ as shown in the figure above.
Recall that each $s$-$t$ path in $G$ corresponds to a clause in $\phi$.}
\label{newfig5:example}
\end{figure}
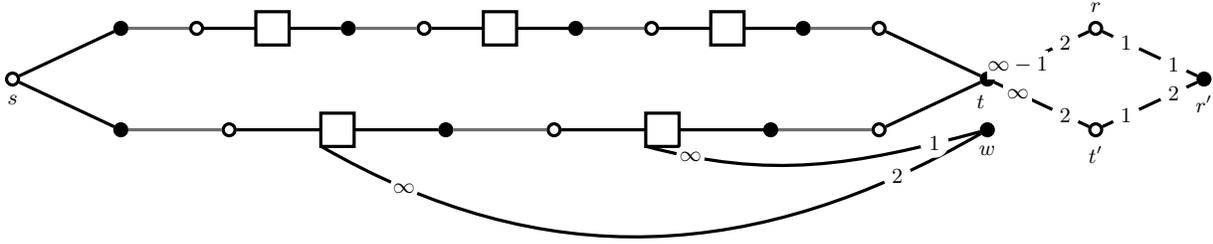

The preference lists of the vertices in $\{r,r',t',t\}$ are as follows:
\[r: r' \succ t \ \ \ \ \ \ \ \ \ \ \ \ \ \ r': r \succ t'\ \ \ \ \ \ \ \ \ \ \ \ \ \ t': r' \succ t\ \ \ \ \ \ \ \ \ \ \ \ \ \ t: \cdots \succ r \succ t'.\]
Recall that the vertex $t$ is adjacent to the last $v^{\ell}$-vertex in every clause gadget $\ell$. The vertex $t$ prefers all its $v$-neighbors (the order among these is not important) to its neighbors in the square which are $r$ and $t'$ and $t$ prefers $r$ to $t'$.

Regarding the vertex $w$, this vertex is the last choice for all its neighbors and $w$'s preference list is some permutation of
its neighbors (the order among these neighbors is not important).
This finishes the description of the graph $H$. We denote the collection of all preference lists again by~${\cal P}$. The proof of Lemma~\ref{lemma:min-size} is analogous to the proof of Lemma~\ref{lem:unstable-edges}.
\begin{lemma}
\label{lemma:min-size}
Let $S$ be any stable matching in $H,{\cal P}$. Then $S$ contains the edges $(u^{\ell}_i,v^{\ell}_i)$ for all clauses $\ell$ and all $i$ and the pair of edges
$(r,r'),(t,t')$. Moreover, $S$ does not contain any consistency edge.
\end{lemma}

\begin{corollary}
\label{cor3}
For any stable matching $S$ in $H, \mathcal{P}$ and any variable $x \in \{X_1,\ldots,X_{2n}\}$, we have:
\begin{enumerate}
  \item From the gadget of $\neg x$, either the pair (i)~$(c_k,d_k),(c'_k,d'_k)$ or (ii)~$(c_k,d'_k),(c'_k,d_k)$ is in $S$.
  \item From a gadget of $x$, say its gadget in the $i$-th clause, either the pair of edges (i)~$(a_i,b_i),(a'_i,b'_i)$ or 
    (ii)~$(a_i,b'_i),(a'_i,b_i)$ is in $S$.
\end{enumerate}
\end{corollary}

  It follows from Lemma~\ref{lemma:min-size} and Corollary~\ref{cor3} that a stable matching in $H$
  matches all vertices except $s$ and $w$. A max-size popular matching in $H$ is a perfect matching:
  it includes the edges $(s,u^{\ell}_0), (v^{\ell}_0,c_k),(d'_k,w)$, $(d_k,c'_k)$ for some negative clause $\ell$. 
  We now prove the following theorem.
\begin{theorem}
  \label{thm:min-size}
  $H,{\cal P}$ admits an unstable min-size popular matching if and only if $G$ admits a stable matching that is dominant.
\end{theorem}  

\begin{proof}
  Suppose $G$ admits a stable matching $N$ that is also dominant. We claim that $M = N \cup \{(r,t),(r',t')\}$
  is a popular matching in $H$. Note that there is a blocking edge $(r,r')$ with respect to~$M$. Since $M$ matches exactly the stable vertices,
  this would make $M$ an unstable min-size popular matching.

  Recall Theorem~\ref{thm:char-popular} (from Section~\ref{prelims}) that characterizes popular matchings. 
  We will now show that the matching $M$ satisfies the three sufficient conditions for popularity as given in Theorem~\ref{thm:char-popular}. 
  Note that $(r,r')$ is the only blocking edge with respect to $M$. 
  Thus property~(ii) from Theorem~\ref{thm:char-popular} obviously
  holds. Property~(i) holds since the edge $(t,t')$ is labeled $(-,-)$ with respect to $M$. 
  Thus there is no alternating cycle in $H_M$ with the edge
  $(r,r')$. We will now show property~(iii) also holds in $H_M$, thus $M$ is a popular matching in $H$.
  
  There are two unmatched vertices in $M$: $s$ and $w$. We need to check that the edge $(r,r')$ is not reachable via an
  $M$-alternating path from either $s$ or $w$ in  $H_M$.  
  Since the matching $N$ is dominant in $G$, there is no
  $M$-alternating path between $s$ and $t$ in $H_M$. Thus the blocking edge $(r,r')$ is not reachable from $s$ by an $M$-alternating path in $H_M$.

  Consider the vertex $w$ and any of its neighbors, say $d'_k$ (see Fig.~\ref{newfig4:example}). We know that $N$ includes either the dotted red pair of edges $(c_k,d'_k),(c'_k,d_k)$ or the dashed blue pair of edges $(c_k,d_k),(c'_k,d'_k)$. In both cases, the blocking edge $(r,r')$
  is not reachable in $H_M$ by an $M$-alternating path with $(w,d'_k)$ as a starting edge. This proves one side of the reduction.

  \medskip

  We will now show the converse. Let $M$ be a min-size popular matching in $H$ that is unstable. Since $M$ is a min-size popular matching,
  the set of vertices matched in $M$ is the set of stable vertices. 
  Consider the edges $(v^{\ell}_{i-1},a_i), (d_k,u^{\ell}_k)$, and
  $(b_i,c_k)$ for any $\ell$, $i$, and $k$. For each of those edges, there is a stable matching in $G$ (and thus in $H$) 
  where all these edges are labeled $(-,-)$ (see the proof of Lemma~\ref{lem:unstable-edges}). Thus these edges are {\em slack}\footnote{An edge $(x,y)$ is slack with respect to a popular matching $N$ and its witness $\vec{\alpha}$ if $\alpha_x + \alpha_y > \wt_N(x,y)$.} with respect to 
  a popular matching and its witness, i.e. $\vec{0}$ here---so they are {\em unpopular} edges (see Section~\ref{prelims}). Hence $M$ has to contain 
  the edges $(u^{\ell}_i,v^{\ell}_i)$ for all clauses $\ell$ and all $i$.

  The 4 vertices in a literal gadget are stable, hence matched among themselves in $M$ (since they must all be matched, and we excluded other edges incident to them), i.e., either the dotted red pair or the dashed blue pair from each literal
  gadget belongs to $M$. Also, the consistency
  edge cannot be a {\em blocking edge} to $M$ as this would make a blocking edge reachable via an $M$-alternating path in $H_M$ from
  the unmatched vertex $w$, a contradiction to $M$'s popularity in $H$ (by Theorem~\ref{thm:char-popular}).

  Since $M$ is unstable, there is a blocking edge with respect to $M$. The only possibility is from within the square $\langle r,r',t',t\rangle$.
  Thus $M$ has to contain the pair of edges $(r,t)$ and $(r',t')$: this makes $(r,r')$ a blocking edge with respect to $M$. Consider the matching
  $N = M \setminus \{(r,t),(r',t')\}$. We have already argued that $N$ is a stable matching in $G$.

  Suppose $N$ is not dominant in $G$. Then there is an $N$-alternating path between $s$ and $t$ in $G_{N}$. Thus in the graph
  $H_M$, there is an $M$-alternating path from the unmatched vertex $s$
  to the blocking edge $(r,r')$. This contradicts the popularity of $M$ in $H$ by Theorem~\ref{thm:char-popular}. 
  Hence $N$ is a dominant matching in $G$. \qed
\end{proof}

We proved in Section~\ref{sec:stable-dominant} that the problem of deciding if $G$ admits a stable matching that is dominant is NP-hard.
Thus we can conclude the following theorem.
\begin{theorem}
  \label{thm:min-size-unstable}
Given $G = (A \cup B,E)$ \new{and $\mathcal{P}$, }
it is NP-complete to decide if $G$ admits a min-size popular matching that is not stable.
\end{theorem}

Our \new{graph} $H$ is such that every popular matching here has size either $n/2$ or $n/2 - 1$, where $n$ is the number of vertices in~$H$. All popular matchings of size $n/2$ are dominant since they are perfect matchings. Thus a popular matching $M$ in $H$ is neither stable nor dominant if and only if $M$ is an unstable popular matching of size $n/2 -1$, i.e., $M$ is an unstable min-size popular matching in~$H$. So we have shown a new and simple proof of NP-hardness of the problem of deciding if a marriage instance admits a popular matching that is neither stable nor dominant.

\subsection{A simple proof of NP-hardness of the popular roommates problem}
\label{sec:pop-roommates}

We know that the popular roommates problem is NP-hard~\cite{FKPZ18,GMSZ18}. Here we adapt the hardness reduction given in Section~\ref{sec:stable-dominant} to show a short and simple proof of hardness of this problem. We first mention some useful structural results from Section~\ref{prelims} that extend to the roommates case.

The first is the {\em Rural Hospitals Theorem}, see \cite[Theorem 4.5.2]{GI89}. That is, if $H = (V,E)$ admits stable matchings, then all stable matchings in $H$ have to match the same subset of vertices. Another useful property is that if an edge $e$ is labeled $(-,-)$ with respect to a stable matching in $H$ then no popular
matching in $H$ can include $e$~\cite{Kav18}. This is based on the {\em tightness} of popular edges in roommates instances and the fact that an edge labeled $(-,-)$ with respect to a stable matching is slack.
Last, we recall that the characterization of popular matchings given in Theorem~\ref{thm:char-popular} also holds when $G$ is non-bipartite.

\begin{pr}
    \inp A (non-bipartite) graph $H = (V,E)$ with strict preference lists.\\
	\ques If there is a popular matching in $H$.
\end{pr}

Given a 3SAT formula $\phi$, we transform it as described in Section~\ref{sec:stable} and build the graph $G$ as described in
Section~\ref{sec:stable-dominant}. We now augment the bipartite graph $G$ into a non-bipartite graph $H$, depicted in Fig.~\ref{fig:roommates} as follows:
\begin{itemize}
\item Add edges between $s$ and the $d'$-vertex in $\neg x$'s gadget for every variable $x$.
\item At the other end of the graph $G$, add an edge $(t,r)$ along with a triangle $\langle r,r',r''\rangle$, where $r,r',r''$ are new
  vertices.
\end{itemize}

\begin{figure}[ht]
	\centering
		\begin{tikzpicture}[scale=0.87, transform shape, very thick]
		\pgfmathsetmacro{\b}{1.6}
		\pgfmathsetmacro{\d}{1.5}
		
		\node[uvertex, label=left:$s$] (s) at (1*\b, \d*0.5) {};
		\node[uvertex] (u0) at (2*\b, 0) {};
		\node[uvertex] (v0) at (3*\b, 0) {};
		\node[uvertex] (u1) at (5*\b, 0) {};
		\node[uvertex] (v1) at (6*\b, 0) {};
		\node[uvertex] (u2) at (8*\b, 0) {};
		\node[uvertex] (v2) at (9*\b, 0) {};
		
		
		\node[uvertex] (u0') at (2*\b, \d) {};
		\node[uvertex] (v0') at (2.7*\b, \d) {};
		\node[uvertex] (u1') at (4.1*\b, \d) {};
		\node[uvertex] (v1') at (4.8*\b, \d) {};
		\node[uvertex] (u2') at (6.2*\b, \d) {};
		\node[uvertex] (v2') at (6.9*\b, \d) {};
		\node[uvertex] (u3') at (8.3*\b, \d) {};
		\node[uvertex] (v3') at (9*\b, \d) {};
		
		\node[uvertex, label=below:$t$] (t) at (10*\b, \d*0.5) {};
				
		\draw (s) -- (u0);
		\draw [MyPurple] (u0) -- (v0);
		\draw (v0) -- node[edgebox] {} (u1);
		\draw [MyPurple] (u1) -- (v1);
		\draw (v1) -- node[edgebox] {} (u2);
		\draw [MyPurple] (u2) -- (v2);
		\draw (v2) --  (t);
		
		
		\draw (s) -- (u0');
		\draw [MyPurple] (u0') -- (v0');
		\draw (v0') -- node[edgebox] {} (u1');
		\draw [MyPurple] (u1') -- (v1');
		\draw (v1') -- node[edgebox] {} (u2');
		\draw [MyPurple] (u2') -- (v2');
		\draw (v2') -- node[edgebox] {} (u3');
		\draw [MyPurple] (u3') -- (v3');
		\draw (v3') -- (t);
		
		\node[uvertex, label=right:$r''$] (r'') at (12*\b, 0) {};
		\node[uvertex, label=below:$r$] (r) at (11*\b, \d*0.5) {};
		\node[uvertex, label=right:$r'$] (r') at (12*\b, \d) {};
		
		\draw (t) -- node[edgelabelr,  near start] {$\infty$} node[edgelabel,  near end] {3} (r);
		\draw (r) -- node[edgelabel,  near start] {1} node[edgelabel,  near end] {2} (r');
		\draw (r') -- node[edgelabel,  near start] {1} node[edgelabel,  near end] {2} (r'');
		\draw (r'') -- node[edgelabel,  near start] {1} node[edgelabelr,  near end] {2} (r);
		
		\draw (s) to[out=-80,in=-150] node[edgelabel, very near end] {$\infty$} ($(v0)!0.5!(u1) +(-0.25,-0.24)$);
		\draw (s) to[out=-90,in=-165] node[edgelabel, very near end] {$\infty$} ($(v1)!0.5!(u2) +(-0.26,-0.24)$);
	\end{tikzpicture}
\caption{We add a new triangle $\langle r,r',r''\rangle$ to $t$, and connect $s$ with all $d'$ vertices as shown in the figure above.}
\label{fig:roommates}
\end{figure}
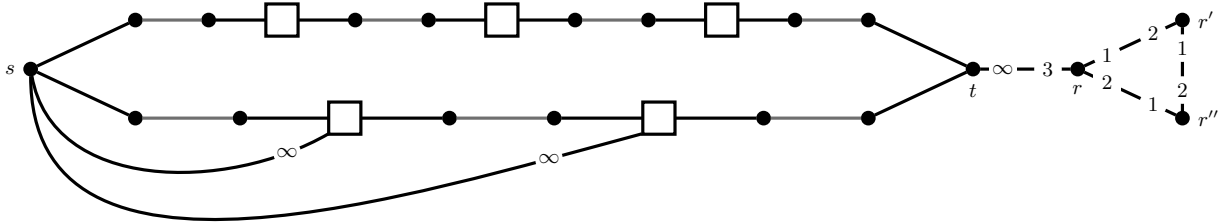
The preference lists of the vertices in $\{r,r',r'',t\}$ are as follows:
\[r: r' \succ r'' \succ t \ \ \ \ \ \ \ \ \ \ \ \ \ \ r': r'' \succ r\ \ \ \ \ \ \ \ \ \ \ \ \ \ r'': r \succ r'\ \ \ \ \ \ \ \ \ \ \ \ \ \ t: \cdots \succ r.\]
The vertex $t$ is adjacent to the last $v^{\ell}$-vertex in every clause gadget $\ell$. The vertex $t$ prefers all its $v$-neighbors 
(the order among these is not important) to $r$. We denote the collection of all preference lists again by ${\cal P}$.

The vertex $s$ is the last choice neighbor for all its neighbors. The preference list of $s$ is not relevant.
Recall the vertex $w$ used in Section~\ref{sec:min-size}: now we have merged the vertices $s$ and $w$. This creates odd cycles, which is allowed here since the graph $H$ is non-bipartite.

\subsubsection{Popular matchings in $H$}
Let $M$ be a popular matching in $H$. Observe that $M$ has to match $r,r'$, and $r''$ since each of these vertices is a top choice neighbor for some vertex. If one of these vertices is left unmatched in $M$ then there would be a blocking edge incident to an unmatched vertex, a 
contradiction to its popularity~(see Theorem~\ref{thm:char-popular}, condition~(iii)). 
Since $t$ is the only {\em outside} neighbor of $r,r',r''$, the matching $M$ has to contain the pair of edges $(t,r),(r',r'')$. Note that the edge $(r,r'')$ blocks $M$.

Let $H^* = H \setminus \{t,r,r',r''\}$. Consider the matching $N = M\setminus\{(t,r),(r',r'')\}$ in $H^*$. Since $M$ is popular in $H$, the
matching $N$ has to be popular in $H^*$.  We claim $N$ has to match all vertices in $H^*$ except $s$. This is because $H^*$ admits a stable
matching: consider $S = \cup_{\ell,i}\{(u^{\ell}_i,v^{\ell}_i)\} \cup\{$dashed blue edges in every literal gadget$\}$---this is a stable
matching in $H^*$. We know that $N$ has to match all stable vertices~\cite{HK13b}. Since the number of vertices in $H^*$ is odd, the vertex $s$ is left unmatched in $N$.

Consider any consistency edge. We claim this is an {\em unpopular} edge in $H^*$. This is because there is a stable matching in $H^*$ where this edge is labeled $(-,-)$
(see the proof of Lemma~\ref{lem:unstable-edges}), hence this edge cannot be used in any popular matching in $H^*$~\cite{Kav18}. Since all vertices in $H^*$ except $s$ have to be matched in $N$, the following 3 observations hold:
 \begin{enumerate}
  \item $N$ contains the edges $(u^{\ell}_i,v^{\ell}_i)$ for all clauses $\ell$ and all $i$.
  \item From the gadget of $\neg x$, either the pair (i)~$(c_k,d_k),(c'_k,d'_k)$ or (ii)~$(c_k,d'_k),(c'_k,d_k)$ is in $N$.
  \item From a gadget of $x$, say, its gadget in the $i$-th clause, either the pair of edges (i)~$(a_i,b_i),(a'_i,b'_i)$ or 
    (ii)~$(a_i,b'_i),(a'_i,b_i)$ is in $N$.
  \end{enumerate}

We are now ready to prove Theorem~\ref{thm:roommates}, whose proof is similar to the proof of Theorem~\ref{thm:min-size}.
\begin{theorem}
  \label{thm:roommates}
  $H,{\cal P}$ admits a popular matching if and only if $G$ admits a stable matching that is dominant.
\end{theorem}
\begin{proof}
  Suppose $G$ admits a stable matching $S$ that is also dominant. We claim that $M = S \cup \{(t,r),(r',r'')\}$ is a popular matching in $H$.  We will again use Theorem~\ref{thm:char-popular} to prove the popularity of $M$ in $H$.
  There is exactly one blocking edge with respect to $M$: this is the edge $(r,r'')$.
  
  Observe that properties~(i) and (ii) from Theorem~\ref{thm:char-popular} are easily seen to hold. 
  We will now show that property~(iii) from Theorem~\ref{thm:char-popular} also holds.
  We need to check that the edge $(r,r'')$ is not reachable via an $M$-alternating path from $s$ in  $H_M$. 
  
  Since the matching $S$ is dominant in $G$, there is no $S$-alternating path between $s$ and $t$ in $G_{S}$. Thus the blocking edge $(r,r'')$ is not reachable from $s$ by an $M$-alternating path in $G_{M}$.
  We now need to show that 
  the blocking edge $(r,r'')$ is not reachable from $s$ by an $M$-alternating path in $H_M$, i.e., when the first edge of the alternating path is $(s,d'_k)$ for some $d'_k$. We know that $M$ includes either the dotted red pair of edges $(c_k,d'_k),(c'_k,d_k)$ or the dashed blue pair of edges $(c_k,d_k),(c'_k,d'_k)$ from this gadget. In both cases, it can be checked that the blocking edge $(r,r')$ is not reachable in $H_M$ by an $M$-alternating path with $(s,d'_k)$ as the starting edge. This proves one side of the reduction.

  \medskip

  We will now show the converse. Suppose $H$ admits a popular matching $M$. We argued above that the pair of edges $(t,r)$ and $(r',r'')$ is in $M$. Consider the matching $N = M \setminus \{(t,r),(r',r'')\}$. We claim that $N$ is a stable matching in $G$.

  From observations~1-3 given above, it follows that the only possibility of a blocking edge to $N$ is from a consistency edge. However a consistency edge cannot block $N$ as this would make a blocking edge reachable by an $M$-alternating path in $H_M$ from the unmatched vertex $s$ and this contradicts $M$'s popularity (by Theorem~\ref{thm:char-popular}).

  We now claim that $N$ is a dominant matching in  $G$. Suppose not. Then there is an $N$-alternating path between $s$ and $t$ in $G_{N}$. Thus in the graph $H_M$, there is an $M$-alternating path from the unmatched vertex $s$ to the blocking edge $(r,r'')$. This contradicts the popularity of $M$ in $H$ by Theorem~\ref{thm:char-popular}. Hence $N$ is a dominant matching in $G$. \qed
\end{proof}  

We know that the problem of deciding if $G$ admits a stable matching that is dominant is NP-hard. This completes our new
proof of the NP-hardness of the popular roommates problem. Observe that every popular matching in $H$ is a max-size matching (since only the vertex $s$ is left unmatched), hence every popular matching in $H$ is dominant. 
Thus our reduction above also shows a simple proof of NP-hardness of the dominant roommates problem.

\subsubsection{Conclusions and an open problem} We considered popular matching problems in a bipartite graph $G = (A \cup B, E)$ with strict preferences. We showed a simple $O(|E|^2)$ algorithm for deciding if there exists a popular matching in $G$ that is {\em not} stable. An open problem is to improve the running time of this algorithm.
We showed that the problems of deciding if a bipartite graph admits a stable matching that is (i)~dominant, (ii)~{\em not} dominant are NP-hard. These results imply many new hardness results for popular matchings in bipartite graphs, including the hardness of finding (1)~a popular matching that is {\em not} dominant, (2)~a min-size popular matching that is {\em not} stable, and (3)~a max-size popular matching that is {\em not} dominant.

\subsubsection{Acknowledgements.} \'{A}gnes Cseh was supported by  the Federal Ministry of Education and Research of Germany in the framework of KI-LAB-ITSE (project number 01IS19066), OTKA grant K128611, and COST Action CA16228 European Network for Game Theory. Telikepalli Kavitha was supported by project no. RTI4001 of the Department of Atomic Energy, Government of India.

\bibliographystyle{abbrv} 
\bibliography{mybib}
\end{document}